\documentclass[a4paper, amsfonts, amssymb, amsmath, reprint, showkeys, nofootinbib, twoside, floatfix,superscriptaddress]{revtex4-1}
\usepackage[english]{babel}
\usepackage{braket}

\usepackage{amsthm}
\usepackage{mathtools}
\usepackage{physics}
\usepackage{xcolor}
\usepackage{graphicx}
\usepackage[left=23mm,right=13mm,top=35mm,columnsep=15pt]{geometry} 
\usepackage{adjustbox}
\usepackage{placeins}
\usepackage[T1]{fontenc}
\usepackage{lipsum}
\usepackage{csquotes}
\usepackage{amsmath,amssymb}
\usepackage{graphicx}
\usepackage{tabularx,booktabs}
\usepackage{tikz}
\usepackage{float}
\usetikzlibrary{quantikz2}
\usepackage[utf8]{inputenc}
\usepackage[colorinlistoftodos, color=green!40, prependcaption]{todonotes}
\usepackage[pdftex, pdftitle={Article}, pdfauthor={Author}]{hyperref}
\usepackage[ruled,lined]{algorithm2e}
\usepackage{algpseudocode}
\usepackage{hyperref}

\usepackage{caption}

\newtheorem{theorem}{Theorem}

\newtheorem{lemma}{Lemma}
\newtheorem{corollary}{Corollary}

\theoremstyle{definition}
\newtheorem{definition}[theorem]{Definition}

\usetikzlibrary{shapes.geometric, arrows}
\usetikzlibrary{automata,positioning}

\newcommand{\vect}[1]{\boldsymbol{#1}}
\newcommand{\ve}[1]{\boldsymbol{#1}}

\tikzstyle{startstop} = [rectangle, rounded corners, minimum width=3cm, minimum height=1cm, text centered, draw=black, fill=blue!20]
\tikzstyle{io} = [trapezium, trapezium left angle=70, trapezium right angle=110, minimum width=3cm, minimum height=1cm, text centered, draw=black, fill=blue!30]
\tikzstyle{process} = [rectangle, minimum width=3cm, minimum height=1cm, text centered, draw=black, fill=blue!20]
\tikzstyle{decision} = [diamond, minimum width=3cm, minimum height=1cm, text centered, draw=black, fill= green!30]
\tikzstyle{arrow} = [thick, ->, >=stealth]

\begin{document}

\title{Adiabatic quantum computing with parameterized quantum circuits}
\author{Ioannis Kolotouros}
    \email{i.kolotouros@sms.ed.ac.uk}
    \affiliation{University of Edinburgh, School of Informatics, EH8 9AB Edinburgh, United Kingdom}
\author{Ioannis Petrongonas}
    \email{ip2004@hw.ac.uk}
    \affiliation{Heriot-Watt University, EH14 4AS Edinburgh, United Kingdom}
\author{Miloš Prokop}
    \email{m.prokop@sms.ed.ac.uk}
    \affiliation{University of Edinburgh, School of Informatics, EH8 9AB Edinburgh, United Kingdom}
\author{Petros Wallden}
    \email{petros.wallden@ed.ac.uk}
    \affiliation{University of Edinburgh, School of Informatics, EH8 9AB Edinburgh, United Kingdom}
\date{\today}

\begin{abstract}

Adiabatic quantum computing is a universal model for quantum computing whose implementation using a gate-based quantum computer requires depths that are unreachable in the early fault-tolerant era.
To mitigate the limitations of near-term devices, a number of hybrid approaches have been pursued in which a parameterized quantum circuit prepares and measures quantum states and a classical optimization algorithm minimizes an objective function that encompasses the solution to the problem of interest. In this work, we propose a different approach starting by analyzing how a small perturbation of a Hamiltonian affects the parameters that minimize the energy within a family of parameterized quantum states. We derive a set of equations that allow us to compute the new minimum by solving a constrained linear system of equations that is obtained from measuring a series of observables on the unperturbed system. We then propose a discrete version of adiabatic quantum computing that can be implemented in a near-term device while at the same time is insensitive to the initialization of the parameters and to other limitations hindered in the optimization part of variational quantum algorithms. We compare our proposed algorithm with the Variational Quantum Eigensolver on two classical optimization problems, namely MaxCut and Number Partitioning, and on a quantum-spin configuration problem, the Transverse-Field Ising Chain model, and confirm that our approach demonstrates superior performance.

\end{abstract}

\maketitle

\section{Introduction}

We are currently transversing from the ``Noisy Intermediate-Scale Quantum devices'' (NISQ) era \cite{preskill2018quantum} to the early fault-tolerant era where small quantum computers of size $\mathcal{O}(100)$ are beginning to achieve useful results \cite{kim2023evidence, da2024demonstration}. Quantum computers promise to offer computational advantage for several tasks ranging from classical optimization \cite{farhi2014quantum}, quantum compiling \cite{khatri2019quantum}, machine learning \cite{cerezo2022challenges}, quantum chemistry \cite{kandala2017hardware} and simulation of quantum mechanical systems \cite{chen2024adaptive}. However, the practicality of these devices is still confined due to noise, the small number of qubits, and the limited connectivity.

A possible way to attempt to circumvent these constraints is to use Variational Quantum Algorithms (VQAs) \cite{bharti2022noisy, cerezo2021variational}, which can reduce the overall circuit depth. In VQAs, a quantum computer works in parallel with a classical computer in a continuous feedback loop in order to minimize an objective function that encompasses the solution to the problem of interest. For the optimization to be successful, the parameterized family of gates must contain the quantum state corresponding to the solution of the problem at hand (or at least a good approximation to it), and the classical optimization algorithm should be able to converge to a near-optimal solution and avoid sub-optimal local minima.

Despite the vast number of applications and research interest, the true performance of these algorithms and whether they can provide a valuable advantage over their classical counterparts is still an open question. A major reason is that the emerging objective function landscapes of VQAs are filled with a high number of local minima which make the algorithms (for specific choices of Ansatz families) $\mathsf{NP}$-hard to train \cite{bittel2021training}. Furthermore, even finding a parameterized family of gates that will contain the solution of the problem is hard. To add to the aforementioned problems, highly expressive Ansatz families that span a large fraction of the total Hilbert space require exponential resources to train due to \emph{barren plateaux} \cite{ragone2023unified, mcclean2018barren, cerezo2021cost}, as the gradients of the objective function vanish exponentially fast as the number of qubits increases. 

To further understand the geometry and trainability of the underlying non-convex landscapes and subsequently understand the limitations of these algorithms, a new field called Quantum Landscape Theory (QLT) was introduced ~\cite{arrasmith2022equivalence, lee2021progress, kim2022quantum, larocca2023theory, kiani2020learning, anschuetz2021critical, koczor2022quantum, wang2024can}. In~\cite{lee2021progress} the authors proved, for the Max-Cut problem, that unless $\mathsf{P}=\mathsf{NP}$, there does not exist an Ansatz family (consisting of commuting generators) with a number of parameters which is polynomial in the system size, and a convex landscape. Furthermore,~\emph{overparametrization} \cite{larocca2023theory, kiani2020learning} for specific Ansatz families can occur if the number of parameters is polynomial to the system size, and may lead to a computational phase transition where only high-quality minima exist in the objective function landscapes.
Moreover, \cite{larocca2022diagnosing} showed that analyzing the Dynamical Lie Algebra of the generators of the Ansatz (and thus the controllability of the system) could help in designing trainable Ansatz families. In this paper, we aim to contribute to the field of QLT by quantifying how much the global minima of these objective functions are shifted if we introduce a small perturbation in the initial Hamiltonian. This information, as we discuss later, has special importance as it can be used as the basis of a hybrid quantum/classical algorithm.

On the other end, conventional techniques for quantum computing with guaranteed performance, such as \emph{Adiabatic Quantum Computing} (AQC) \cite{albash2018adiabatic} require depth (or else coherent evolution for large time interval) that is unreachable for current quantum devices. In AQC the system is initialized in the ground state of an easy-to-compute ground state and the Hamiltonian is ``slowly''  varied until it becomes the Hamiltonian of interest. However, the system must be varied sufficiently slowly so that it remains in the instantaneous ground state throughout the evolution. Then, at the final time $t_f$, the system will be found in the ground state of the desired Hamiltonian. The time taken to complete the evolution quantifies the cost/resources required for a given computation. What determines the minimum time that suffices for the problem to be solved (i.e. how to ensure adiabaticity) is the spectral gap  \cite{amin2009consistency}. The total evolution time $t_f$ must scale as the inverse of the spectral gap, meaning that problems in which the gap becomes exponentially small \cite{van2001powerful} require exponentially large time and thus cannot be efficiently solved. Moreover, addressing the effect of noise in the adiabatic quantum evolution is also a complicated task.

In the past few years, the notions and ideas of AQC have tried to be incorporated into the NISQ literature \cite{keever2023towards, garcia2018addressing, harwood2022improving, chen2020demonstration}. In \cite{garcia2018addressing, harwood2022improving} the authors incorporated certain ideas from AQC into the Variational Quantum Eigensolver. Specifically, they defined a (discrete) parameterized Hamiltonian similar to the one used in AQC, and they started from a Hamiltonian with a known ground state that belongs to a certain ansatz family of parametrized states. Then, they iteratively tried to minimize the expectation value of the (parameterized) Hamiltonian by using the output parameters at every step as the starting point/initialization for the parameters of the next. Then they argued that the final output would be close to the optimal angles and the whole procedure could be used as a warm-starting method \cite{egger2021warm}. Relevant to this work, \cite{chen2020demonstration} proposed a method to variationally simulate the adiabatic evolution, and \cite{chandarana2022digitized} proposed to enhance the Quantum Approximate Optimization Algorithm (QAOA) with additional counterdiabatic driving terms which were shown to outperform the standard QAOA for the problems they investigated. Finally, \cite{hibat2021variational} used recurrent neural networks to simulate an annealing framework and showed that on average their method outperforms simulated annealing on several spin-glass problems.

While the approach of \cite{garcia2018addressing, harwood2022improving} seems to offer a potential advantage over traditional VQE methods (since bad initialization of VQE may lead to far-from-optimal minima) it still seems to under-perform in certain problems, even for small instances. Specifically, these methods rely on the classical optimization of a time-evolving objective function, and so the problems of local minima and false convergence still persist \footnote{Note that a potential challenge that this approach faces is that a small change in the Hamiltonian may result in the previous optimal point transforming into a \textit{saddle point} where any gradient-based optimization algorithm would fail.} \footnote{A Newton's type update that considers the Hessian may point towards a direction that increases the energy since it has both positive and negative eigenvalues.}. In addition, their method requires fine-tuning both the hyperparameters of the classical optimizer as well as the step, without quantifying how much the latter may affect the position of the global minimum. 

Inspired by AQC, we consider an optimization algorithm that is significantly different from VQAs and the work mentioned above. In our proposal, we start with a parameterized Hamiltonian (similar to AQC) and initialize the Ansatz family with angles that minimize the initial Hamiltonian. Then, at each iteration we \emph{calculate} the optimal angles \emph{analytically}, using the expectation values of certain observables at the previous step and solving a system of linear equations. Unlike the other works, we do \emph{not} run an energy minimization at each step, and we are thus less affected by the landscape of the cost space (e.g. barren plateaux). At the end of the algorithm, we output the angles that minimize the Hamiltonian of interest.

Our work draws connection to \emph{predictor-corrector} methods that are widely used in classical optimization \cite{simonetto2016class, simonetto2020time}. In these methods, the goal is to approximate the global minimum of a time-dependent cost function. A predictor (usually an Euler step) identifies a position near the global minimum and the corrector (a classical optimization algorithm) corrects the position of the predictor by minimizing (locally) the cost function of the problem. However, these techniques are used for convex optimization problems and so they cannot be directly applied in the VQAs framework.\\

\noindent\emph{Our Contributions:}
\begin{itemize}
    \item We study how small perturbations of a Hamiltonian affect the optimization landscape and specifically how much the minima are shifted under these perturbations. This enables us to follow the trajectory of a minimum in the cost landscape, as a (parameterized) Hamiltonian varies by solving a constrained linear system of equations. As such we obtain the ground state of the varying Hamiltonian without relying on a correct choice of hyperparameters at each perturbation.
    
    \item We formulate an algorithm to find the best approximation of the ground state of a Hamiltonian within a family of parameterized quantum states that: (i) can be applied in an early fault-tolerant device, (ii) is not sensitive to the initialization points and (iii) requires fixed calls to the quantum computer with theoretical guarantees on the performance.    

    \item We quantify the quantum resources needed for our algorithm, and show that finding the minimum of the perturbed Hamiltonian can be cast as a semidefinite program.
    
    \item We test our proposal on (simulated) small instances of both classical optimization problems such as Max-Cut and Number Partitioning, and on ``non-classical'' Hamiltonians which include non-diagonal terms, such as the random Transverse-Field Ising Chain. In both cases, the solution is achieved to a high accuracy within a few steps.

    \item We evaluate our proposed method by comparing it with the Variational Quantum Eigensolver on the same problems. 
    
\end{itemize}

\noindent\emph{Structure:} In Section \ref{sec:preliminaries} we give the essential background on Parameterized Quantum Circuits, Variational Quantum Algorithms and Adiabatic Quantum Computing. In Section \ref{sec:var_adiab} we present our main theoretical results and formulate an optimization algorithm that is inspired by Adiabatic Quantum Computing and can be used by NISQ/early-fault tolerant devices that produce high fidelity parameterized Quantum Circuits. In Section \ref{sec:param_theory} we give the proof of the main mathematical theorem that determines how much the global minima of the emerging objective function landscapes are shifted when a small perturbation in the Hamiltonian is introduced. In Section \ref{sec:solving_the_linear_system}, we provide an algorithm to solve the main problem incorporated in our approach, and quantify the number of quantum resources required to solve the problem. In Section \ref{sec:experiments} we test our approach by performing quantum simulations on the Max-Cut problem and on Transverse Ising Chains. In Section \ref{sec:comparisons} we compare our method with the Variational Quantum Eigensolver on two classical optimization problems and on a quantum spin-configuration problem. We conclude in Section \ref{sec:conclusion} with a general discussion of our results and future work.

\section{Preliminaries}
\label{sec:preliminaries}

We briefly introduce the notions of parameterized Quantum Circuits and their properties, Variational Quantum Algorithms and Adiabatic Quantum Computing that will be used throughout the paper.

\subsection{Parameterized Quantum Circuits}
\label{subsection:param_circuits}

Consider a quantum circuit composed of a series of parameterized and non-parameterized gates. These \emph{parameterized quantum circuits} (PQCs) can be described by a unitary operator $U(\boldsymbol{\theta})$, with $\boldsymbol{\theta} = (\theta_1,\ldots, \theta_M)=\sum_{i=1}^M\theta_i \ve{\hat{e}_i}$ being the parameter vector \footnote{In order for the Ansatz to be trainable, the number of parameters $M$ must scale as $M=\mathcal{O}(poly(n))$, where $n$ is the system size.} and $\hat{e}_i$ being the unit vector pointing in the $i$th direction. The type/number of gates is usually referred to as the \emph{Ansatz family}. If $K$ is the total number of gates of the PQC, then the Ansatz family can be written in its more generic form as
\begin{equation}
    U(\ve{\theta}) = \prod_{k=K}^1 e^{-i\theta_k g_k},
\label{eq:ansatz_family}
\end{equation}
where $g_k$ are the generators ($g_k^\dagger = g_k)$ corresponding to every unitary of the PQC. We use the convention that the gate parameterized with angle $\theta_1$ acts first. In most cases only a small fraction of the $K$ total parameters are tunable, while the other $K-M$ angles correspond to fixed non-parameterized gates (usually two-qubit gates introducing entanglement). This type of Ansatz family corresponds to many PQCs that appear in the literature, such as the hardware-efficient Ansatz \cite{kandala2017hardware}, the quantum alternate operator Ansatz \cite{hadfield2019quantum}, the variational Hamiltonian Ansatz \cite{wecker2015progress} or the quantum optimal control Ansatz \cite{choquette2021quantum}. The dimension of the parameter space is $M$, and depends on the PQC considered, but typically depends on the number of qubits $n$ that are used, which in its turn depends on the size of the instance of the problem one attempts to solve. For example typical hardware efficient Ans\"atze have $M=O(n)$ to $O(n^2)$, while QAOA may perform well with $M=O(\log n)$.

The classical subroutines in Variational Quantum Algorithms (see Subsection \ref{subsection_var_alg}) require knowing the derivatives of the expectation value of a Hamiltonian with accurate precision. Specifically, consider the expectation value/energy of a Hamiltonian $H$ on the state $\ket{\psi(\vect{\theta})} = U(\boldsymbol{\theta})\ket{0}$,
\begin{equation}
    F(\boldsymbol{\theta}) = \bra{0}U^\dagger (\boldsymbol{\theta}) H U(\boldsymbol{\theta})\ket{0}.
\end{equation}
Let a variable $\theta_j$ correspond to a parameterized gate in the circuit with a generator $g_j$, i.e. $V(\theta_j) = e^{-i\theta_j g_j}$. The derivative of the expectation value can be approximated with high precision using the \emph{finite-differences} approximation as
\begin{equation}
    \frac{\partial F}{\partial \theta_j} \approx \frac{F(\boldsymbol{\theta} + h \hat{e}_j) - F(\boldsymbol{\theta} - h \hat{e}_j)}{2h},
\label{eq:finite_differences}
\end{equation}
with Eq \eqref{eq:finite_differences} becoming exact as $h\rightarrow 0$. The Hessian $\textbf{H}$ of the energy $F(\vect{\theta})$ is defined via its matrix elements, namely $\mathbf{H}_{jk}=\frac{\partial^2 }{\partial \theta_j\partial \theta_k}F(\vect{\theta})$. The finite difference approximation can also be applied for higher-order derivatives, e.g. for the second order derivatives of the Hessian, we have
\begin{equation}
\begin{gathered}
    \frac{\partial^2 F}{\partial \theta_j \partial \theta_k} \approx \frac{1}{4h^2}\Big[F(\boldsymbol{\theta} + h \hat{e}_j + h \hat{e}_k) - F(\boldsymbol{\theta} - h \hat{e}_j + h \hat{e}_k)\\
    -F(\boldsymbol{\theta} + h \hat{e}_j -h \hat{e}_k) + F(\boldsymbol{\theta} - h \hat{e}_j -h \hat{e}_k)\Big].
\end{gathered}
\end{equation}

Recently, methods that allow for the exact calculation of the derivatives with the same computational overhead as the finite-differences approximation have been developed. If the generator $g_j$ has two distinct eigenvalues $\pm r$, then the \emph{parameter-shift rules} \cite{romero2018strategies,schuld2019evaluating, mari2020estimating} state that the derivative $\frac{\partial F}{\partial \theta_j}$ can be calculated exactly as
\begin{equation}
    \frac{\partial F}{\partial \theta_j} = r\left[F\left(\boldsymbol{\theta}+\frac{\pi}{4r}\boldsymbol{\hat{e}_j}\right) - F\left(\boldsymbol{\theta}-\frac{\pi}{4r}\boldsymbol{\hat{e}_j}\right) \right].
\label{eq:par_shift}
\end{equation}
 As a result, shifting the parameters in the appropriate direction allows for exact calculation of the derivatives compared to finite-difference methods. Moreover, the parameter-shift rules can be used for higher-order derivatives as well.  Since its elements $\textbf{H}_{jk}$ for the two parameters $\theta_j$, $\theta_k$ correspond to two parameterized gates with generators having eigenvalues $\pm r$ it can be calculated as
\begin{equation}
\begin{gathered}
    \textbf{H}_{jk} = r^2\Big[F\left(\boldsymbol{\theta}+\frac{\pi}{4r}(\boldsymbol{\hat{e}_j} + \boldsymbol{\hat{e}_k})\right) - F\left(\boldsymbol{\theta}+\frac{\pi}{4r}(\boldsymbol{\hat{e}_j} -\boldsymbol{\hat{e}_k})\right)\\
    -F\left(\boldsymbol{\theta}+\frac{\pi}{4r}(-\boldsymbol{\hat{e}_j} +\boldsymbol{\hat{e}_k})\right) + F\left(\boldsymbol{\theta}-\frac{\pi}{4r}(\boldsymbol{\hat{e}_j} +\boldsymbol{\hat{e}_k})\right) \Big].
\label{eq:par-hessian}
\end{gathered}
\end{equation}

\subsection{Variational Quantum Algorithms}
\label{subsection_var_alg}

\newsavebox{\genericfilt}
\savebox{\genericfilt}{%
\begin{tikzpicture}
        \node[scale=0.75] {
            \begin{quantikz}
            \lstick{$\ket{0}^{\otimes n}$} &\gate{V(\theta_1)}\qwbundle[alternate]{} &\gate{V(\theta_2)}\qwbundle[alternate]{} &\ldots & &\gate{V(\theta_M)}\qwbundle[alternate]{}  &\meter{}\qwbundle[alternate]{}
            \end{quantikz}
        };
    \end{tikzpicture}
}
\begin{figure*}
\includegraphics[width=17.5cm,height=10cm,keepaspectratio]{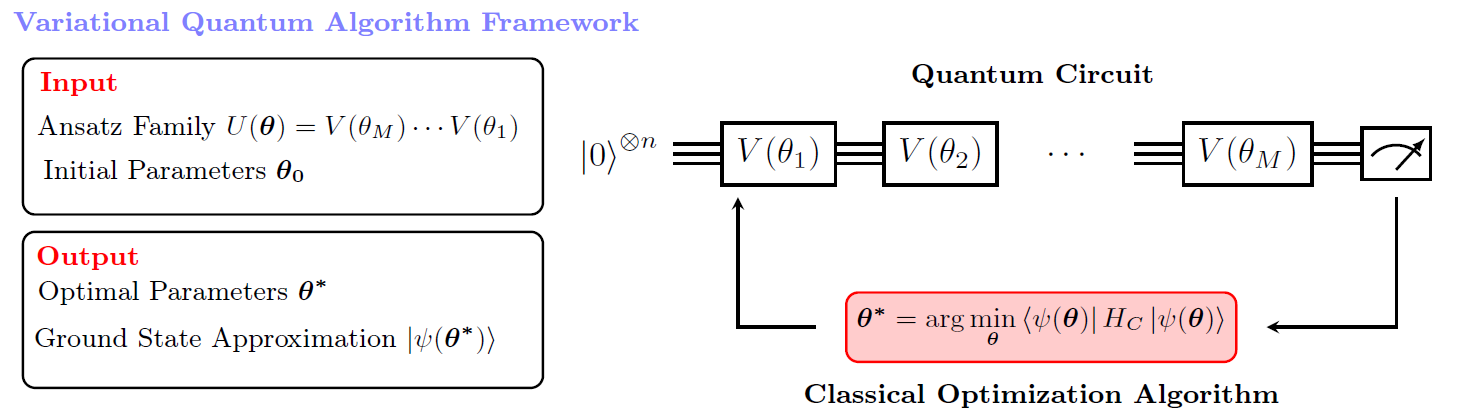}
\caption{General Variational Quantum Algorithm Framework. In the first step, the parameterized family of gates and the initial parameters are chosen. Then, the classical optimization algorithm iteratively updates the parameters towards the direction of the ground state energy. When the classical optimizer converges to a minimum, the algorithm stops and an approximation to the ground state is given at the output.}
\label{fig:quantum_circuit}
\end{figure*}

Variational Quantum Computing is a hybrid framework in which a quantum computer works in parallel with a classical computer in order to solve a given task. It is a heuristic algorithm and thus we do not have guarantees that the solution it outputs is optimal. The general framework is outlined in Figure \ref{fig:quantum_circuit}.

The first step is to express the problem at hand as an interacting qubit Hamiltonian $H_C$ whose ground state will correspond to the optimal solution of the problem. This is usually done by decomposing the Hamiltonian $H_C$ into $L=\mathcal{O}(poly(n))$ Pauli strings (with $n$ being the problem size), namely
\begin{equation}
H_C = \sum_{l=1}^L c_l P_l,
\end{equation}
where $c_l$ is the real coefficient corresponding to the $l$-th Pauli string $P_l$.

Next, an \emph{Ansatz family} $U(\boldsymbol{\theta})$ parameterized by a vector $\boldsymbol{\theta} = (\theta_1, \theta_2, \ldots, \theta_M)$ with $M=\mathcal{O}(poly(n))$ parameters is chosen so that the ground state (or at least a good approximation of it) is included within the Ansatz. The parameterized family of gates then acts on an easily-prepared reference state $\ket{\phi}$, usually taken to be the $\ket{0}^{\otimes n}$, and produces the state $\ket{\psi(\boldsymbol{\theta})}$, i.e. $U(\boldsymbol{\theta})\ket{0}^{\otimes n} = \ket{\psi(\boldsymbol{\theta})}$. The success of a variational quantum algorithm depends strongly on the choice of the parameterized gates, and therefore the design of highly expressible and easily trainable Ansatz families has received significant interest \cite{grimsley2019adaptive, kandala2017hardware, hadfield2019quantum, choquette2021quantum, zhu2022adaptive}. 

In the third step, once the parameterized state $\ket{\psi(\boldsymbol{\theta})}$ has been prepared, the circuit is executed and the output state is measured. The quantum computer works in parallel with the classical computer in order to minimize an objective function (whose minimum corresponds to the solution of the problem), usually chosen to be the expectation value of the output state. However, in many cases different objective functions lead to improved performance \cite{barkoutsos2020improving, kolotouros2022evolving, li2020quantum}. When the classical optimization algorithm has converged, the variational quantum algorithm terminates and outputs the optimal angles $\boldsymbol{\theta^*}$ such that, when VQA is successful,
\begin{equation}
\boldsymbol{\theta^*} = \arg\underset{\boldsymbol{\theta}}{\min}\bra{\psi(\boldsymbol{\theta})}H_C\ket{\psi(\boldsymbol{\theta})}.
\end{equation}

\subsection{Adiabatic Quantum Computing}

Adiabatic Quantum Computing (AQC) seeks to evolve a state under a time-dependent Hamiltonian $H(t)$. More specifically, a system of qubits is initialized into an easy-to-prepare ground state of a Hamiltonian $H_0$. Then, the system is allowed to interact through the Hamiltonian
\begin{equation}
    H(t) = \left(1-\frac{t}{t_f}\right)H_0 + \frac{t}{t_f}H_1, \text{ $t\in [0,t_f]$}.
\label{eq:adiabatic_hamiltonian}
\end{equation}
If the Hamiltonian is gapped and the evolution is ``slow enough'' so that the system of qubits always remains in the instantaneous ground state throughout the evolution, then at the final time $t_f$ the system will be in the ground state of the desired Hamiltonian $H_1$. There are exact bounds to restrict the time $t_f$ in order to ensure adiabaticity \cite{jansen2007bounds, elgart2012note} and is correlated to the spectral gap, i.e. the energy difference between the ground state and the first excited state. Specifically, exponentially (with the system size) small gaps require exponentially large evolution time the Hamiltonian \eqref{eq:adiabatic_hamiltonian}.

The building block (and inspiration) of AQC is the \emph{Adiabatic Theorem} which states that an evolving quantum system under a time-dependent Hamiltonian $H(t)$ will remain in the instantaneous ground state as long as the system never ``receives'' enough energy to make a transition to the instantaneous first excited state. The sufficient energy to make a transition is bounded by the spectral gap. As a result, all AQC-inspired algorithms must have a runtime that is dependent on the minimum spectral gap of the evolution. There are problems, such as solving linear systems of equations or search-engine problems, where bounding the spectral gap is possible. In \cite{costa2022optimal}, \emph{Costa et al.} were able to show that a discrete version of AQC can achieve an asymptotically optimal scaling for solving linear systems while in \cite{garnerone2012adiabatic} the authors were able to provide a polylogarithmic (to the system size) AQC algorithm for the PageRank problem.

As proven in \cite{van2001powerful}, the total unitary evolution $U(t_f, 0) = e^{-i\int_0^{t_f}H(t)dt}$ required to interpolate between the two Hamiltonians can be approximated by $M = \mathcal{O}(\text{poly}(n)t_f)$ discrete unitary evolutions where $n$ is the system size. Clearly, the large depth that is required to approximate the adiabatic evolution makes it intractable for near-term devices as the system size increases. Thus it would be useful to examine whether we can use a trade-off between trainable parameters and circuit depth.

\subsection{Notation}
 At this point, it is important to clarify our notation, and specifically to highlight the different Hamiltonians used throughout the manuscript.

\begin{itemize}
    \item $H_0$: The initial Hamiltonian at time $t=0$, for which we know the parameters that produce its ground state. For all our experiments, we choose $H_0 = -\sum_{j} \sigma_j^x$ with a ground state $\ket{\psi_0} = \ket{+}^{\otimes n}$.

    \item $H_1$: This is the target Hamiltonian that the user aims to find its ground state. The goal is to identify the parameters of the parameterized quantum circuit that generates the ground state of $H_1$.

    \item $V$: This Hamiltonian corresponds to the perturbation that we add in each iteration. In the case of Algorithm \ref{alg:param_algorithm}, we choose the perturbation Hamiltonian to be $V \equiv H_1 - H_0$.

    \item $H_s$: This is the starting Hamiltonian \emph{at each step of the algorithm} (or else the unperturbed Hamiltonian). It is equivalent with $H_0$ at $t=0$, but it could be any intermediate Hamiltonian for which we have found its ground state (in the previous step) using Theorem \ref{th:main}.

    \item $H_{\lambda}$: This is the perturbed Hamiltonian on all intermediate steps (or else the Hamiltonian that we seek its ground state on intermediate steps). That is, on every iteration, we start with the ground state of $H_s$ and we use Theorem \ref{th:main} to find the ground state of $H_\lambda = H_s + \lambda V$. $H_\lambda$ is also equivalent to $H_1$ at time $t=t_f$ when the algorithm terminates.
\end{itemize}

\section{Adiabatic Quantum Computing with Parameterized Quantum Circuits}
\label{sec:var_adiab}

In this section we will present our two main results: A theorem that quantifies how small changes in the Hamiltonian affect the position of the global minima, and a hybrid quantum-classical algorithm, which we call AQC-PQC, that utilizes the aforementioned theorem and ideas from AQC to find the best approximation of the ground state of a Hamiltonian within a family of quantum states obtained using a parameterized quantum circuit. A comparison with other approaches is given at the end of the section while the proof of the theorem follows in Section \ref{sec:param_theory}. 

Consider a Hamiltonian $H_s$, whose ground state, just like in AQC, is known. Let also a parameterized quantum circuit $U(\boldsymbol{\theta})$, prepared with parameters $\boldsymbol{\theta^*}$ that generate the ground state $\ket{\psi(\boldsymbol{\theta^*})}$ of $H_s$. The first question that we want to answer is: ``If the Hamiltonian $H_s$ is deformed by a small amount $\lambda V$ ($H_\lambda = H_s + \lambda V$), what is the shift vector $\boldsymbol{\epsilon}$ that will translate the system from the initial ground state $\ket{\psi(\boldsymbol{\theta^*})}$ of $H_s$ to the ground state $\ket{\psi(\boldsymbol{\theta^*} + \boldsymbol{\epsilon})}$ of the slightly deformed Hamiltonian $H_\lambda$?". The answer to this question is given in Theorem \ref{th:main}.

\begin{figure}
\includegraphics[scale=0.65]{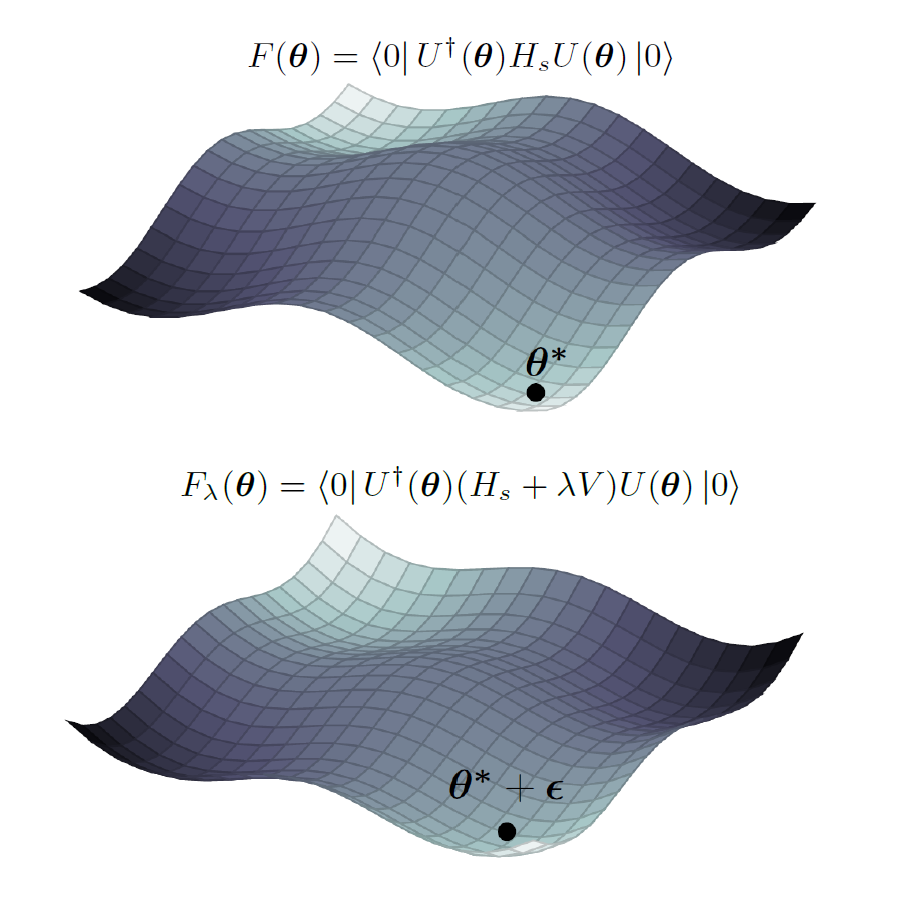}
    \caption{Variation of the loss function landscape for a small variation in the Hamiltonian. The global minimum shifts from the optimal point $\boldsymbol{\theta^*}$ of the Hamiltonian $H_s$ to the optimal point $\boldsymbol{\theta^*}+\boldsymbol{\epsilon}$} of the Hamiltonian $H_s + \lambda V$.
\end{figure}

\begin{theorem}
\label{th:main}
Consider a \emph{parameterized quantum circuit} defined via the unitaries $U(\boldsymbol{\theta})$, and the corresponding states $\ket{\psi(\boldsymbol{\theta})}=U(\boldsymbol{\theta})\ket{0}$. We are given a Hamiltonian $H_s$ and the angles $\boldsymbol{\theta^*}$ that minimize its energy, i.e. $\boldsymbol{\theta^*} = \arg \min_{\boldsymbol{\theta}} \bra{\psi(\boldsymbol{\theta})}H_s\ket{\psi(\boldsymbol{\theta})}$. If we perturb the Hamiltonian $H_s$ by a small amount $\lambda V$ with $\lambda \ll 1$, and $\Vert H_s\Vert\approx\Vert V\Vert $, then there exists a shift vector $\ve{\epsilon}$ such that, with high probability, the state $\ket{\psi(\boldsymbol{\theta^*}+\boldsymbol{\epsilon})}$ is the ground state of the perturbed Hamiltonian  $H_\lambda = H_s + \lambda V$  and the shift vector is the solution of the following mathematical problem: 
\begin{equation}
\label{eq:main_problem}
\begin{gathered}
        \text{min } \norm{\boldsymbol{\epsilon}}\\
        \text{subject}\text{ to: } A\boldsymbol{\epsilon} + \ve{Q} = 0,\\
        \textbf{H}^\lambda\big{|}_{\boldsymbol{\theta^*}+\boldsymbol{\epsilon}} \succcurlyeq 0,
\end{gathered}
\end{equation}
where $\ve{H}^\lambda\big{|}_{\ve{\theta^*+\epsilon}}$ is the Hessian evaluated at the shifted point, $\ve{Q}=\sum_i Q_i \ve{\hat{e}_i}$ is a vector and $A$ is a matrix that are defined via their elements
\begin{eqnarray}
\label{eq:Q_A}
Q_i  &=& \lambda  \frac{\partial}{\partial\theta_i} \left(\bra{\psi(\ve{\theta})} V\ket{\psi(\ve{\theta})}\right) \bigg{|}_{\ve{\theta^*}}\\
A_{ij} &=&  \frac{\partial^2}{\partial\theta_i\partial\theta_j} \left(\bra{\psi(\ve{\theta})} H_\lambda\ket{\psi(\ve{\theta})}\right)\bigg{|}_{\ve{\theta^*}}.\nonumber 
\end{eqnarray}
\end{theorem}
Intuitively, we are looking for the smaller shift (first line of Eq. \ref{eq:main_problem}) that has vanishing gradient (second line) which at the same time is a minimum (rather than saddle point or maximum), as given by the constraint in the Hessian matrix (third line). An important point is to note that the elements $A_{ij}, Q_j$ and the Hessian matrix can all be calculated using expectations and derivatives for the unperturbed state  $\ket{\psi(\boldsymbol{\theta^*})}$. 
 This means that if we use this approach iteratively (see below), to compute the ``new ground state'', one needs to prepare a fixed number of quantum states in each step\footnote{The details depend on the PQC used, e.g. on whether parameter-shift rules are applicable or not.}.
 
 The motivation behind Theorem \ref{th:main} is outlined below. As we discuss next, one could use the aforementioned theorem in an iterative manner and construct an algorithm that allows the ground state preparation of a target Hamiltonian. Consider the task of finding the ground state of \emph{a target Hamiltonian $H_1$}. The user would utilize a parameterized quantum circuit and initialize the parameters in the ground state of \emph{a different but known Hamiltonian $H_0$}. If the Hamiltonian is deformed sufficiently slowly (by introducing small perturbations) and at the final time the Hamiltonian $H_0$ has transformed into the target Hamiltonian $H_1$, then by iteratively applying Theorem \ref{th:main} (after each small deformation) the user would reach the target ground state (or an approximation of it). Note, however, that the known ground state (of the Hamiltonian $H_s$) at each iteration is the ground state of the slightly deformed Hamiltonian of the previous iteration. We can now come back to the aim of the paper, to obtain a method that uses PQC to approximate AQC. Our approach can be summarized in Algorithm \ref{alg:param_algorithm}.

\begin{algorithm}[H]
\caption{Adiabatic Quantum Computing with Parameterized Quantum Circuits}
\label{alg:param_algorithm}
\SetKwInOut{Input}{Input}
\Input{Initial Hamiltonian $H_0$\;
Target Hamiltonian $H_1$\;
Ansatz family $\ket{\psi(\boldsymbol{\theta})}=U(\boldsymbol{\theta})\ket{0}$ with $M$ parameters such that the ground state of $H_0$ and $H_1$ (or a good approximation of them) is contained within the ansatz\;
$\boldsymbol{\theta^*} = \arg\underset{\boldsymbol{\theta}}{\min}\bra{\psi(\boldsymbol{\theta})}H_0\ket{\psi(\boldsymbol{\theta})}$\;
Set of expectation values of observables: the Hessian $\textbf{H}^\lambda$ and $\{Q_i, A_{ij}\}$ as given in Eq. (\ref{eq:Q_A})\;
Total steps $K$;}
\For{$k=1,2,\ldots, K$}{
$H_k = (1-\frac{k}{K})H_0 + \frac{k}{K}H_1$\;
Measure and estimate $\{Q_i, A_{ij}, \textbf{H}^\lambda\}$ using a quantum processor\;
Use Eq \ref{eq:main_problem} and a classical solver to calculate $\ve{\epsilon}=\left(\epsilon_1, \epsilon_2, \ldots, \epsilon_M\right)$\;
{$\ve{\theta^*}=\ve{\theta^*+\epsilon}$\;}}
\Return $\ket{\psi(\boldsymbol{\theta^*})}$
\end{algorithm} 

\begin{figure*}
\includegraphics[width=17.5cm,height=10cm,keepaspectratio]{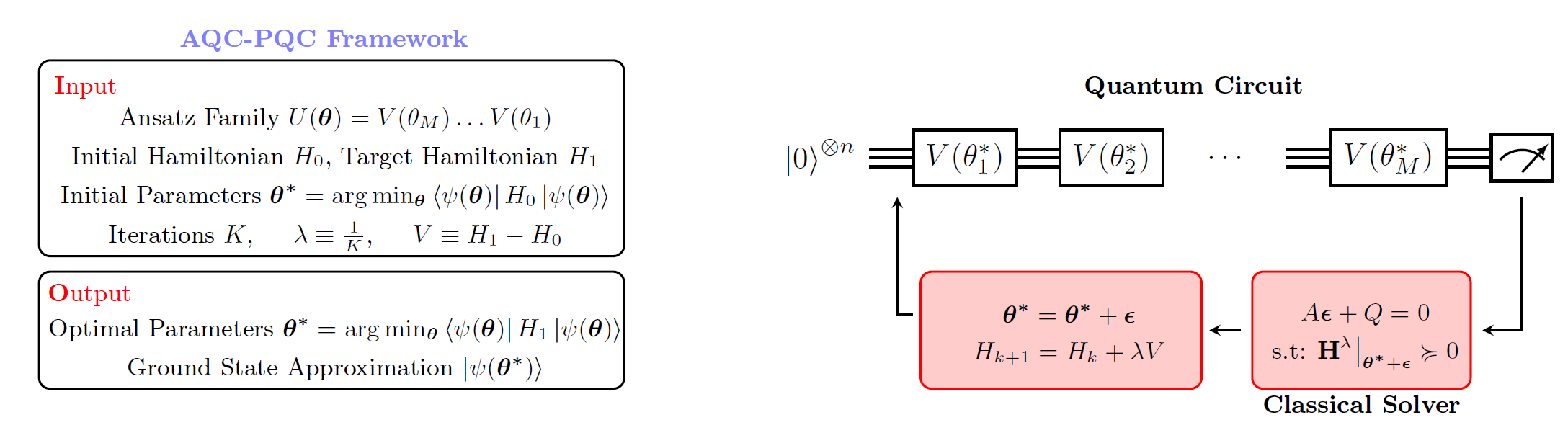} 
\caption{Adiabatic Quantum Computing with Parameterized Quantum Circuits. At every iteration, a series of observables are calculated for the instantaneous ground state. These observables form a linear system of equations whose solution corresponds to the shift vector $\boldsymbol{\epsilon}$. After the solution is found, the parameters $\boldsymbol{\theta^*}$ are shifted by $\boldsymbol{\epsilon}$ and the ground state is updated to $\ket{\psi(\boldsymbol{\theta^*}+\boldsymbol{\epsilon})}$. Then, a small perturbation is added to the Hamiltonian in order for the new observables to be calculated. Finally, a ground state approximation of $H_1$ is given at the output.}
\label{fig:loop_variational}
\end{figure*}

The algorithm can also be visualized in Figure \ref{fig:loop_variational}.  The main idea is the following. We start with discretizing AQC in a way similar to VAQC \cite{garcia2018addressing, harwood2022improving}, taking $K$ steps. We consider the (step-dependent) Hamiltonian 
 \begin{equation}
     H_k = \left(1-\frac{k}{K}\right)H_0 + \frac{k}{K}H_1.
 \end{equation} 
 Here the step subscript $k$ has the role of the (discrete in our case) time. Let us set $\lambda:=1/K \ll 1$, and $V:=(H_1-H_0)$. We can rewrite the step-dependent Hamiltonian as
 \begin{equation}
     H_k=H_0+ k\lambda V .
 \end{equation}
 We can easily see that $H_{k+1}-H_k=\lambda V$, and thus we can apply Theorem \ref{th:main} for any consecutive pair of  $\{H_k,H_{k+1}\}$. We start from $H_s$ and initialize the algorithm with the known ground state that corresponds to the initial parameters $\boldsymbol{\theta^*}$. Then for each step, we compute the shift vector $\boldsymbol{\epsilon}$ and add it to the parameters corresponding to the ground state of the previous step to obtain the ground state of the next step. For example, at the $k+1$ iteration, we can apply Theorem \ref{th:main} using $H_k$ as the known Hamiltonian $H_s$ and $\lambda (H_1-H_0)$ as the perturbation. Thus, by solving Eq. \eqref{eq:main_problem} we can calculate the shift vector $\boldsymbol{\epsilon}$ that will translate the system onto the ground state of the Hamiltonian $H_{k+1}$.
 
 As we noted after our main theorem, to compute the shift vector we need (i) to estimate the expectation values of certain observables evaluated for the starting state using our quantum device and (ii) use a classical solver to solve Eq. \ref{eq:main_problem}. The shifted ground state is then used as the ground state for the next step. In the $K$th (final) step, the Hamiltonian becomes $H_1$ and thus the ground state we recover is the desired ground state terminating the quantum/classical loop. \\

 \noindent\textbf{Comparisons.} Our method differs significantly from both the traditional adiabatic evolution and VAQC proposed by \cite{garcia2018addressing,harwood2022improving}. In AQC the adiabatic evolution, when run in a digital quantum computing device, is approximated by a series of Trotterized unitary evolutions. As a result, in order to simulate these unitaries, a circuit with a very large depth is needed which makes AQC inapplicable for near-term devices. In our approach, we take advantage of the fixed architecture of the parameterized quantum circuits, which supports quantum states that can be prepared with relatively high fidelity, while we still exploit the advantages and guarantees that adiabatic quantum computing offers (at least in the noiseless case). Similarly to variational algorithms, we delegate part of the computation to a classical processor to construct a hybrid algorithm. In AQC the (time-evolving) Hamiltonian needs to be physically implemented (or a Trotterized version of it). In our approach, we evaluate the energy corresponding to that Hamiltonian by measuring the parameterized quantum states we produce, which allows us to easily consider more general Hamiltonians (for example ones that include terms of higher order than quadratic).

Our approach also differs significantly from variational quantum algorithms, offering a number of potential advantages. The main difference in our approach is that we do not perform energy minimization in the conventional way. As such, we do not have to rely on empirical choices of the hyperparameters for the success of our method as we do not iteratively follow the direction of the negative gradient. In our approach, the energy minimization is performed by finding the closest minimum, which is identified by solving the constrained linear system in Eq \eqref{eq:main_problem}.

As we discussed, the classical part is a constrained linear solver (which can be performed very efficiently and with guarantees of finding the solution), while in variational approaches the user usually selects a first or second-order minimization method. Most of the limitations of VQAs come exactly from the optimization part and specifically due to two main bottlenecks. First of all, a random initialization of parameters may lead to either bad performance or \emph{barren plateaux}. Secondly, the emerging landscapes of VQAs are filled with a vast amount of local minima or barren plateaux that make them untrainable. While there exist methods that try to overcome these limitations \cite{sack2021quantum, jain2022graph, egger2020quantum, park2024hamiltonian} our algorithm offers a more robust strategy (see also for simple simulated experiments in Section \ref{sec:experiments}). 

Another key difference, and advantage of our approach, is that traditional variational quantum approaches require constant preparation of the quantum circuits. For each iteration, multiple quantum states need to be prepared (details depend on the classical optimizer used). As the classical optimization algorithm approaches the minimum, the number of shots and quantum state preparations increase significantly (as the gradients tend to zero). On top of that, we neither know in advance how many iterations would be required to reach convergence nor if the quantum state that we converge to is the correct ground state\footnote{It is a heuristic approach.}. In contrast, in our strategy, we can mimic adiabatic quantum computing with only $K$ steps, where $K$ is the chosen number of discretization steps and is typically much smaller than the iterations a classical optimizer requires in VQAs. Therefore, we have a known number of quantum states that are to be prepared and measured \footnote{We require $\mathcal{O}\Big((1 + \dim(\mathcal{N}_{\kappa}(A))M^2\Big)$ different quantum states for each step if we follow the approach used in Lemma \ref{lemma:qresources}}, and also have a guarantee that our method will find the solution if AQC can solve the problem efficiently and $K$ is chosen suitably. 

A final potential advantage compared to traditional variational approaches, is that quantum error mitigation methods may be more effective in our case. There are QEM methods (for example \cite{bennewitz2022neural}) that are specifically applicable if the quantum state considered is close to the ground state. In our method, all quantum states used are close to the ground state of some time-dependent Hamiltonian, and thus these approaches should be more productive. A full analysis of these implications, as well as the questions stated in the above paragraph, are a subject for further research.

\section{Parameterized Perturbation Theory}
\label{sec:param_theory}

In this section, we will first prove Theorem \ref{th:main}, which reduces the problem of finding how a small perturbation of a Hamiltonian shifts the parameters that minimize the energy, to a constrained system of linear equations. We then give two ways to impose the constraint. Finally, for a special (but very widely used) class of families of parameterized quantum circuits, we give expressions of how to exactly evaluate the required derivatives and outline the method we follow for our simulated experiments in the next section. 

\begin{proof}[Proof of Theorem \ref{th:main}] We first give the outline in four steps.
To establish that we have found a minimum of a function we need to (i) check that the gradient of the function, evaluated at that point, vanishes (critical point) and that (ii) the second derivative (Hessian) corresponds to a positive semi-definite matrix (is a minimum). Our function is the (perturbed) energy $F_\lambda (\ve{\theta})$ corresponding to the Hamiltonian $H_\lambda=H_s+\lambda V$ of the parametrized quantum state $\ket{\psi(\ve{\theta})}$. The first step is to consider the gradient of a general shifted point $\ve{\theta^*}+\ve{\epsilon}$. For the new ground state, this gradient should vanish. However, since our aim is to find the shift vector $\ve{\epsilon}$, we cannot measure the energy of the shifted state until we know (or until we have a good guess for) the shift vector. The second step is to exploit the fact that the Hamiltonian is only slightly perturbed. We expect that the new ground state is close to the one we started, thus the shift vector $\ve{\epsilon}$ is small and we can expand the energy around the previous minimum using a Taylor expansion, while we are justified in keeping only the leading (linear) terms in $\boldsymbol{\epsilon}$. Now all the functions (and derivatives) required are evaluated for the previous (and known) value $\ve{\theta^*}$. The third step is to evaluate the derivatives either approximately using finite differences, or exactly if the PQC allows us to use the parameter shift rules (Eqs \ref{eq:par_shift},\ref{eq:par-hessian}). This reduces the problem to a system of linear equations. This system (in general) has many solutions since some of the equations are not independent. If $\lambda$ is sufficiently small, and for well-behaved Hamiltonians, the smallest shift vector that gives a minimum is the global minimum and thus the ground state. The final fourth step requires exactly this, to minimize over the possible solutions for the shift vector, while also confirming that the solution is indeed a minimum by checking the positivity of the Hessian. \\


\noindent \emph{Step 1.} We consider the energy $F_\lambda (\ve{\theta})$ given by

\begin{equation}
    F_\lambda(\boldsymbol{\theta}) = \bra{0} U^\dagger (\boldsymbol{\theta}) H_\lambda U(\boldsymbol{\theta}) \ket{0}.
\end{equation}
We are interested in the value of the new ground state, i.e. the new minimum of the energy. We expect the new ground state, since the Hamiltonian is perturbed slightly by $\lambda V$, to be close to the previous ground state, i.e. we search for some point $\ve{\theta^*+\epsilon}$. For this point to be a minimum, the gradient of the energy should vanish,
\begin{equation}
\label{eq:gradient1}
    \frac{\partial}{\partial\theta_i}F_\lambda(\ve{\theta})\bigg{|}_{\ve{\theta}=\ve{\theta^*+\epsilon}}=0 \quad \forall \quad i.
\end{equation}
\\
\noindent \emph{Step 2.} We can expand the energy around the old ground state, using a Taylor expansion, as 
\begin{equation}
    F_\lambda (\boldsymbol{\theta^*} + \boldsymbol{\epsilon}) = F_\lambda(\boldsymbol{\theta^*}) + \sum_{i=1}^M \epsilon_i\frac{\partial }{\partial \theta_i} F_\lambda(\ve{\theta^*}) + \mathcal{O}\left(\norm{\boldsymbol{\epsilon}}^2\right).
\label{eq:taylor_Fl}
\end{equation}
For suitably small perturbation (i.e. sufficiently small $\lambda$), the shift vector $\ve{\epsilon}$ is also small and we can approximate accurately the energy by keeping up to the linear in $\boldsymbol{\epsilon}$ term of the Taylor expansion and plug it into Eq. (\ref{eq:gradient1}) to impose a vanishing gradient,
\begin{equation}
\label{eq:main0}
    \frac{\partial}{\partial\theta_i} F_\lambda(\ve{\theta^*})+ \frac{\partial}{\partial\theta_i}\left( \sum_{j=1}^M \epsilon_j\frac{\partial }{\partial \theta_j} F_\lambda(\ve{\theta^*})\right)=0 \quad \forall\quad i.
\end{equation}
By noting that
\[F_\lambda(\ve{\theta^*})=\bra{\psi(\ve{\theta^*})}H_s\ket{\psi(\ve{\theta^*})}+\lambda \bra{\psi(\ve{\theta^*})}V\ket{\psi(\ve{\theta^*})}
\]
and that
\[
\frac{\partial}{\partial \theta_i}\bra{\psi(\ve{\theta^*})}H_s\ket{\psi(\ve{\theta^*})}=0\quad\forall\quad i,
\]
and because $\ve{\theta^*}$ is the ground state of $H_s$, we have
\begin{equation}
\label{eq:Qi1}
    \frac{\partial}{\partial\theta_i} F_\lambda(\ve{\theta^*})=\lambda \frac{\partial}{\partial \theta_i}\bra{\psi(\ve{\theta^*})}V\ket{\psi(\ve{\theta^*})}:= Q_i.
\end{equation}
Similarly, we can define
\begin{equation}
\label{eq:Aij1}
    A_{ij}:=\frac{\partial^2}{\partial\theta_i\partial\theta_j}F_\lambda(\ve{\theta^*}).
\end{equation}
Eq. (\ref{eq:main0}) becomes
\begin{equation}
    Q_i+\sum_j A_{ij}\epsilon_j=0\quad\forall\quad i,
\end{equation}
the system of equations in Eq. (\ref{eq:main_problem}) in Theorem \ref{th:main}.
\\
\noindent \emph{Step 3.} Both the $Q_i$ and the $A_{ij}$ involve derivatives of expectation values evaluated at the known point $\ve{\theta^*}$. As we have seen in Section \ref{sec:preliminaries}, one can evaluate such derivatives by computing the expectation values at a number of points (one or two per dimension of the parameter space) related to $\ve{\theta^*}$. In the general case, those points need to be very close to $\ve{\theta^*}$, and the result is an approximation (finite differences). In the special case that the generators have a certain specific form (see later) one can evaluate exact derivatives using the parameter-shift rule (see Subsection \ref{subsection:param_circuits}). The exact choice depends on the PQC that one considers, suggesting that PQC that admit parameter-shift rules would perform more accurately in our method. \\

\noindent \emph{Step 4.} To determine the (small) shift in the parameter space that the ground state moved because of a small perturbation, we need to find (i) the turning point closest to the old ground state (vanishing determinant) that also (ii) is actually a minimum. To ensure the former we need to take the $\ve{\epsilon}$ with the minimum norm, while for the latter we need to ensure that the Hessian $H^\lambda_{jk}=\frac{\partial^2}{\partial\theta_j\partial\theta_k}F_\lambda(\ve{\theta})$ evaluated at the new point $\ve{\theta^*+\epsilon}$ is a positive semi-definite matrix\footnote{Recently, the
authors of \cite{huembeli2021characterizing} numerically analyzed the objective function landscapes that appear in the highly parameterized vector spaces of variational quantum algorithms. They noted that in most cases, the Hessian matrix at the local and global minima is positive semi-definite, with the
zero eigenvalue being highly degenerate.}. The shift vector is therefore the solution to the problem
\begin{equation}
\begin{gathered}
        \text{min } \norm{\boldsymbol{\epsilon}}\\
        \text{subject}\text{ to: } A\boldsymbol{\epsilon} + \ve{Q} = 0,\\
          \textbf{H}^\lambda\big{|}_{\boldsymbol{\theta^*}+\boldsymbol{\epsilon}} \succcurlyeq 0.
\end{gathered}
\end{equation}
\end{proof}
To solve the constrained problem in Eq. (\ref{eq:main_problem}), both a classical and a quantum device is needed. Ideally, we would like to first measure some expectation values, and then, with these values as input, use a (fully) classical solver to find the shift vector. However, the constraint involves checking the (semi) positivity of a matrix (Hessian) which is evaluated for the new ground state $\ve{\theta^*} + \ve{\epsilon}$. Since the shift vector is not known (yet), one would think that the constraint cannot even be defined unless one has a candidate shift vector. While this approach is possible (see \emph{remark} later), we can also exploit, once more, the fact that the shift vector is small for small perturbations. 
\begin{lemma}
\label{lemma:first}
The mathematical problem of Eq. (\ref{eq:main_problem}) can be solved using expectation values of observables and their derivatives evaluated at the known point $\ve{\theta^*}$. Specifically, the Hessian at $\ve{\theta^*+\epsilon}$ can be approximated using this expression
\begin{equation}
\label{eq:hessian2}
    \textbf{H}^\lambda_{jk}\big{|}_{\ve{\theta=\theta^*+\epsilon}}=\textbf{H}^\lambda_{jk}\big{|}_{\ve{\theta=\theta^*}}+\sum_{i=1}^M \epsilon_i \frac{\partial}{\partial\theta_i}\textbf{H}^\lambda_{jk}\big{|}_{\ve{\theta=\theta^*}}
\end{equation}
\end{lemma}
\begin{proof}
The linear set of equations is already defined at $\ve{\theta^*}$ while the Hessian is obtained by Taylor-expanding and keeping up to linear terms.
\end{proof}
Here we should note that to compute the Hessian in the way described, we require up to third derivatives of the ground state energy. Since each derivative requires us to evaluate the state in at least one different point per dimension of the parameter space, computing third derivatives would require $O(M^3)$ quantum states.\\

\noindent\textbf{Remark.}
An alternative approach would be to use a method that has more cycles of classical-quantum subroutines. Specifically, one could first evaluate expectation values at $\ve{\theta^*}$ and solve the simpler (unconstrained) system
\begin{equation}
    \begin{gathered}
        \text{min } \norm{\boldsymbol{\epsilon}}\\
        \text{subject}\text{ to: } A\boldsymbol{\epsilon} + \ve{Q} = 0.
\label{eq:experiments}
\end{gathered}
\end{equation}
Then, using the trial shift-vector $\ve{\tilde{\epsilon}}$, prepare the new set of states (of $O(M^2)$) to check if the Hessian at the point $\ve{\theta^*}+\ve{\tilde{\epsilon}}$ is positive semi-definite. If it is, one outputs $\ve{\epsilon=\tilde{\epsilon}}$. If it is not, one goes back and finds a new candidate shift-vector and prepares a new set of quantum states to check the Hessian at that point. The process terminates when a suitable solution is found. While this approach may require $O(M^2)$ preparations of quantum states, it has some disadvantages that led us to focus on Lemma \ref{lemma:first}: (i) It is not clear after how many rounds we are guaranteed (or likely) to find a suitable solution; (ii) We need to go back and forth between the classical and quantum processors; (iii) Optimized classical solvers for the constrained problem cannot be used, and a naive, trial-and-error method imposing the constraint is followed. \\

As a final point, we note that for most PQCs, for example, those that have Pauli rotations as parameterized quantum gates, one can use parameter-shift rules to evaluate the derivatives exactly.

\begin{lemma}
\label{th:main2}
Consider the statement in Theorem \ref{th:main}, where the parameterized quantum circuit is defined via unitaries that have generators $g_j$ with two distinct eigenvalues $\pm r$. Then Eq. (\ref{eq:main_problem}) can be evaluated exactly using
\begin{gather}
\label{eq:QA2}
Q_i= \frac{\lambda}{2}\left(V\left(\ve{\theta^*}+\frac{\pi}{2}\ve{\hat{e}_i}\right)-V\left(\ve{\theta^*}-\frac{\pi}{2}\ve{\hat{e}_i})\right)\right)\\
A_{ij}= \frac{1}{4}\left(F_\lambda\left(\ve{\theta^*}+\frac{\pi}{2}\ve{\hat{e}_i}+\frac{\pi}{2}\ve{\hat{e}_j}\right)-F_\lambda\left(\ve{\theta^*}-\frac{\pi}{2}\ve{\hat{e}_i}+\frac{\pi}{2}\ve{\hat{e}_j}\right)\right.\nonumber\\
\left.-F_\lambda\left(\ve{\theta^*}+\frac{\pi}{2}\ve{\hat{e}_i}-\frac{\pi}{2}\ve{\hat{e}_j}\right)+F_\lambda\left(\ve{\theta^*}-\frac{\pi}{2}\ve{\hat{e}_i}-\frac{\pi}{2}\ve{\hat{e}_j}\right)\right)\nonumber
\end{gather}
and similarly, the Hessian in Eq. (\ref{eq:hessian2}) can be evaluated by taking third derivatives with the parameter-shift rule.
\end{lemma}
\begin{proof}
This follows directly using the parameter-shift rules given in Eqns \eqref{eq:par_shift}, \eqref{eq:par-hessian}.
\end{proof}

\section{Solving the constrained linear system}
\label{sec:solving_the_linear_system}

\noindent In this section, we will decompose the main mathematical problem incorporated in AQC-PQC and focus on all of its subsequent parts separately. As it is clear, the hardness in our approach comes in solving the main mathematical problem in Eq. \eqref{eq:main_problem}:

\begin{equation*}
\begin{gathered}
        \text{min } \norm{\boldsymbol{\epsilon}}\\
        \text{subject}\text{ to: } A\boldsymbol{\epsilon} + \boldsymbol{Q} = 0,\\
        \textbf{H}^\lambda\big{|}_{\boldsymbol{\theta^*}+\boldsymbol{\epsilon}} \succcurlyeq 0,
\end{gathered}
\end{equation*}
using the least number of classical and quantum resources.

Recall that we are searching for the minimum vector $\boldsymbol{\epsilon} \in \mathbb{R}^m$ that has zero gradient and is a (global) minimum. The norm of a vector $\norm{\boldsymbol{\epsilon}}$ is a convex function. The same holds for the linear equation constraint, as the affine map $A\boldsymbol{\epsilon} + \boldsymbol{Q}$ is also a convex function. The main bottleneck in our problem is that the Hessian at the point ($\boldsymbol{\theta^*} + \boldsymbol{\epsilon})$ is not (in general) an affine map in $\boldsymbol{\epsilon}$, making it a non-convex problem. However, as we will see, this is not the case when $\boldsymbol{\epsilon}$ is small.\\

\subsection{Equality Constraint}

\noindent We start with the equality constraints:
\begin{equation}
\label{eq:lin_equation}
    A\boldsymbol{\epsilon} + \boldsymbol{Q} = 0
\end{equation}
with $A\in \mathcal{S}^m$ ($\mathcal{S}^m = \{X |\; X^T = X\}$ is the space of symmetric matrices) and $\boldsymbol{Q}\in \mathbb{R}^m$. The first step is to eliminate all \emph{equality constraints}, as in most cases this linear system of equations is overdetermined. Consider any $\boldsymbol{\epsilon_0}$ that is a solution of Eq. \eqref{eq:lin_equation}, i.e. $A\boldsymbol{\epsilon_0} + \boldsymbol{Q} = 0$. It will be wise, for reasons that will become clear later, to choose $\boldsymbol{\epsilon}_0$ to be the smallest vector that satisfies the linear equation, i.e. $\boldsymbol{\epsilon_0}$ is the solution of the convex problem:
\begin{equation*}
\begin{gathered}
        \text{min } \norm{\boldsymbol{\epsilon}}\\
        \text{subject}\text{ to: } A\boldsymbol{\epsilon} + \boldsymbol{Q} = 0
\end{gathered}
\end{equation*}

Such a vector is unique. Let $\mathcal{F}$ be the feasible set of solutions of Eq. \eqref{eq:lin_equation}:
\begin{equation}
    \begin{gathered}
    \mathcal{F} = \{\boldsymbol{\epsilon}\big{|} A\boldsymbol{\epsilon}+ \boldsymbol{Q} = 0\}\implies 
    \mathcal{F} = \{\boldsymbol{\epsilon}| A(\boldsymbol{\epsilon} - \boldsymbol{\epsilon_0}) = 0\} \implies \\
    \mathcal{F} = \{ \boldsymbol{u}+\boldsymbol{\epsilon_0}| A\boldsymbol{u} = 0\}
    \end{gathered}
\end{equation}
where in the last line we defined $\boldsymbol{u} \equiv \boldsymbol{\epsilon}-\boldsymbol{\epsilon_0}$. As a result, all $\boldsymbol{\epsilon} = \boldsymbol{u}+\boldsymbol{\epsilon_0}$, with $\boldsymbol{u} \in \mathcal{N}(A)$ (where $\mathcal{N}(A)$ is the null space of $A$) correspond to solutions of the linear system of equations. Since Eq. \eqref{eq:lin_equation} is a linear equation (whose solution provides us with the critical points in the energy landscape) which is approximately equal to 0, it is more appropriate to define the notion of an $\kappa$-approximate null space (denoted as $\mathcal{N}_{\kappa}(A)$), i.e $A \boldsymbol{u}\approx 0$.\\

\noindent \textbf{Definition 1 ($\kappa$-Approximate null space).} \textit{Let $A\in \mathcal{S}^m$ and let $\kappa > 0$. Consider the set of eigenvectors $\{\boldsymbol{v_1}, \ldots, \boldsymbol{v_l}\}$ of $A$ such that $\norm{\boldsymbol{v_i}} = 1$ and $\norm{A\boldsymbol{v_i}} \leq \kappa$. The $\kappa$-approximate null space $\mathcal{N}_{\kappa}(A)$ is defined as $\mathcal{N}_{\kappa}(A) := \text{span}(\boldsymbol{v_1}, \ldots, \boldsymbol{v_l})$}.\\

\noindent We will now proceed and explain how one can construct an approximate null space. As a first step, we apply a singular value decomposition (SVD) of the matrix $A$. By doing so, the matrix A can be decomposed as:
\begin{equation}
    A = U \Sigma V^T
\label{eq:SVD_decomposition}
\end{equation}
where $U, V$ are orthogonal matrices and $\Sigma$ is a diagonal matrix with singular values of $A$ as its entries. Let $\text{diag}(\Sigma) = (\sigma_1(A)> \sigma_2(A)> \ldots > \sigma_m(A)>0)$ be the singular values of $A$. Since $A$ is symmetric, $U=V$ and their columns are the eigenvectors of $A$. Additionally, the singular values $\sigma_i$ are the absolute values of the eigenvalues of $A$, i.e. $\sigma_i = \abs{\lambda_i}$. If $\boldsymbol{v_i}$ are the eigenvectors of $A$ then from Eq. \eqref{eq:SVD_decomposition}, we can write $A$ as:
\begin{equation}
    A = \sum_{i=1}^m \sigma_i \boldsymbol{v_i} \boldsymbol{v_i}^T
\label{eq:a-decomp}
\end{equation}

As our next step, we can apply a low-rank approximation of the matrix $A$. Specifically, from Eq. \eqref{eq:a-decomp} we can keep all terms up to the $k$-th term. Let $A_k = \sum_{i\leq k} \sigma_i \boldsymbol{v_i} \boldsymbol{v_i}^T$ be the $k$-rank approximation of $A$. The error in the approximation is given in the Frobenius norm as:
\begin{equation}
    \norm{A-A_k}_F = \sqrt{\sum_{k+1}^{m}\sigma_i^2}
\end{equation}
and is the optimal $k$-rank approximation according to \textit{Eckart-Young-Mirsky Theorem}. As a result, the above analysis provides a recipe for how to define the $\kappa$-approximate null space $\mathcal{N}_{\kappa}(A)$. That is, one can set a threshold $\kappa>0$ with $\kappa \approx 0$ so that all singular values smaller than $\kappa$ are neglected. The basis of $\mathcal{N}_{\kappa}(A)$ is then clearly $(\boldsymbol{v_{k+1}}, \ldots, \boldsymbol{v_m}) = \text{span}(\mathcal{N}_{\kappa}(A))$ with $\dim(\mathcal{N}_{\kappa}(A)) = m-k$. Consider now the most general vector $\boldsymbol{\mu} \in \mathcal{N}_{\kappa}(A)$ 
\begin{equation}
    \boldsymbol{\mu} = c_{k+1} \boldsymbol{v_{k+1}} +\ldots + c_m \boldsymbol{v_m} = \sum_{i=k+1}^m c_i \boldsymbol{v_i}
\end{equation}
with $c_i\in \mathbb{R}$. As a result, the main mathematical problem \eqref{eq:main_problem} has thus been reformulated to:
\begin{equation}
\begin{gathered}
    \min_{\boldsymbol{\mu}} ||\boldsymbol{\epsilon_0} + \boldsymbol{\mu}|| \\
    \text{subject}\text{ to: } \textbf{H}^\lambda\big{|}_{\boldsymbol{\theta^*}+\boldsymbol{\epsilon_0} + \boldsymbol{\mu}} \succcurlyeq 0\\
    \boldsymbol{\mu} \in \mathcal{N}_{\kappa}(A)
\label{eq:reformulated_problem}
\end{gathered}
\end{equation}

\begin{corollary}
    Solving the main mathematical problem in Eq. \eqref{eq:main_problem} reduces the classical search to only $\dim(\mathcal{N}_{\kappa}(A)) < m$ parameters.
\end{corollary}

As a result, we can reduce the overall hardness required in gradient approaches, where we usually have to optimize all the parameters of the parameterized quantum circuit. Relevant approaches where only a subset of the total parameters are varied have been previously utilized in \cite{kolotouros2023random, ding2023random}. Ii is worth noting that in problems that we examine in section Sec \ref{sec:experiments}, the dimension of the $\kappa$-approximate null space is usually much smaller than the total number of parameters $m$.

\subsection{Positive-semidefinite constraint}

Since we eliminated the equality constraint in Eq. \eqref{eq:main_problem} we will proceed to satisfy the second constraint which corresponds to the Hessian. Specifically, our goal is to make the Hessian matrix positive-semidefinite ($\textbf{H}^\lambda|_{\boldsymbol{\theta^*}+\boldsymbol{\epsilon}}\succcurlyeq 0$), as the solution of Eq. \eqref{eq:main_problem} must translate the system to the ground state of the perturbed Hamiltonian which is a (global) minimum. Making the Hessian matrix positive-semidefinite is equivalent to making its smallest eigenvalue greater or equal to zero.

Consider a matrix $X\in \mathcal{S}^m$ where $\mathcal{S}^m = \{X |\; X^T = X\}$ is the space of symmetric matrices. Clearly, any Hessian matrix belongs in $\mathcal{S}^m$.

\begin{definition}
    (Minimum eigenvalue function). The function $f:\mathcal{S}^m \rightarrow \mathbb{R}$ that inputs a symmetric matrix $X$ and outputs its minimum eigenvalue is defined as:
    \begin{equation}
        f(X) = \inf_{v}\{\boldsymbol{v}^T X \boldsymbol{v} \; | \; \norm{\boldsymbol{v}} = 1 \}
    \end{equation}
\end{definition}

The minimum eigenvalue function has the following important property, highlighted in Lemma \ref{lemma:minimum_eigenvalue_function}.

\begin{lemma}
    The function $f:\mathcal{S}^m \rightarrow \mathbb{R}$ that inputs a symmetric matrix $X$ and outputs its minimum eigenvalue is concave. If $X_1, X_2 \in \mathcal{S}^m$ and $0\leq \theta \leq 1$, then $f$ satisfies Jensen's inequality:
    \begin{equation}
    f[\theta X_1 + (1-\theta) X_2] \geq \theta f(X_1) + (1-\theta) f(X_2)
    \end{equation}
    \label{lemma:minimum_eigenvalue_function}
\end{lemma}

\begin{proof}
    The proof follows immediately, as $f$ is the pointwise infimum of a family of linear functions (i.e. $\boldsymbol{v}^T X \boldsymbol{v}$). For more details, see \cite{boyd2004convex}.
\end{proof}

\begin{definition}
    We define the composite function $h = f \circ \textbf{H}: \mathbb{R}^m \rightarrow S^m \rightarrow \mathbb{R}$:
    \begin{equation}
        h(\boldsymbol{\epsilon}) = f(\textbf{H}^\lambda|_{\boldsymbol{\theta^*}+ \boldsymbol{\epsilon}})
    \end{equation}
\end{definition}

It is clear that the function $h$ is not concave since the Hessian operator is not in general affine in $\boldsymbol{\epsilon}$. However, from our previous analysis, we have assumed that for small perturbations $\lambda$, the shift vector that translates the system onto the new ground state is also small. This allows us to define the affine approximation of the Hessian.

\begin{definition}
    The affine approximation $\tilde{\textbf{H}}$ of $\textbf{H}$ is defined as:
    \begin{equation}
        \tilde{\textbf{H}} = \textbf{H}^{\lambda}|_{\boldsymbol{\theta^*}} + \sum_{k=1}^m \epsilon_k D_k|_{\boldsymbol{\theta^*}}
    \end{equation}
    where the matrix $D_k$ is defined as $D_k \equiv \frac{\partial \textbf{H}^\lambda}{\partial \theta_k}$.
\end{definition}

One can quantify the error of the affine approximation of the Hessian from the full matrix as described in Appendix \ref{appendix:error_analysis}.

\begin{lemma}
    The function $h = f\circ \tilde{\textbf{H}}$ defined as:
    \begin{equation}
        h(\boldsymbol{\epsilon}) = f\Big( \textbf{H}^{\lambda}|_{\boldsymbol{\theta^*}} + \sum_{k=1}^m \epsilon_k D_k|_{\boldsymbol{\theta^*}}\Big)
    \end{equation}
    is concave in $\mathbb{R}^m$.
\end{lemma}

\begin{proof}
    The proof follows similarly to Lemma \ref{lemma:minimum_eigenvalue_function}, as $h$ is the pointwise minimum of a family of affine functions.
\end{proof}

Our analysis has allowed us to express the problem as a semidefinite program. It is straightforward to see that if we used the $\kappa$-approximate null space defined in Definition 1 and use the fact that $\boldsymbol{\epsilon} = \boldsymbol{\epsilon_0}+\boldsymbol{\mu}$, where $\boldsymbol{\mu} \in \mathcal{N}_{\kappa}(A)$ then Eq. \eqref{eq:main_problem} can be reduced to its final form:
\begin{equation}
    \begin{gathered}
    \min_{\boldsymbol{\mu}} ||\boldsymbol{\epsilon_0} + \boldsymbol{\mu}|| \\
    \text{subject}\text{ to: } f\Big(\textbf{H}^\lambda +  \grad_{\boldsymbol{\mu}} \textbf{H}^\lambda\Big{|}_{\boldsymbol{\theta^*}+\boldsymbol{\epsilon_0}}\Big) \geq 0\\
    \boldsymbol{\mu} \in \mathcal{N}_{\kappa}(A)
    \label{eq:reformulated_problem2}
    \end{gathered}
\end{equation}
where $\grad_{\boldsymbol{\mu}} = \boldsymbol{\mu}^T \grad$ is the directional derivative pointing in the $\kappa$-approximate null space. As a result, in Lemma \ref{lemma:qresources}, we can conclude what are the essential quantum resources to formulate and solve the mathematical problem defined in Eq. \eqref{eq:main_problem}.

\begin{lemma}
    (Quantum Resources). The total number of quantum resources required at every step of the AQC-PQC algorithm scales as:
    \begin{equation*}
        \mathcal{O}\Big((1 + \dim(\mathcal{N}_{\kappa}(A))M^2\Big)
    \end{equation*}
    where $A$ is the hessian of the perturbed Hamiltonian at the ground state of the unperturbed Hamiltonian.
\label{lemma:qresources}
\end{lemma}
\begin{proof}
    The proof follows immediately from our previous analysis.
\end{proof}

Finally, once the problem has been formulated, classical algorithms, such as interior point methods \cite{boyd2004convex}, or supergradient ascent methods \cite{boyd2003subgradient} can be utilized to return the solution. 


\section{Simulated Experiments}
\label{sec:experiments}
 We will first give an overview of the technical details of the experiments and we will also introduce the different classes of problems that we examined. Then, we will provide an analysis of our method for different choices of discretization steps.

\subsection{Technical Details}
For our experiments, we used both \emph{Qiskit Statevector} and QuEST simulators which allow exact \emph{noiseless} calculation of the expectation values.

For the parameterized family of gates, we chose a \emph{hardware-efficient ansatz} consisting of a layer of $R_y$ rotations for each qubit, followed by a series of controlled-$Z$ operations applied in a nearest-neighbor fashion, and then finally another layer of $R_y$ rotations (see Figure \ref{fig:Ansatz_max_cut}). In order to evaluate the efficiency of our algorithm we used Eq. \eqref{eq:experiments}. Then, to test the condition that the Hessian at the point $\boldsymbol{\theta^*} + \boldsymbol{\epsilon}$ is positive semidefinite, we prepared the state $\ket{\psi(\boldsymbol{\theta^*} + \boldsymbol{\epsilon})}$ and calculated the Hessian using Eq. \eqref{eq:par-hessian}. Below we analyze the different problems that we evaluated our method.

\begin{figure}
    \centering
    \includegraphics[width=0.48\textwidth]{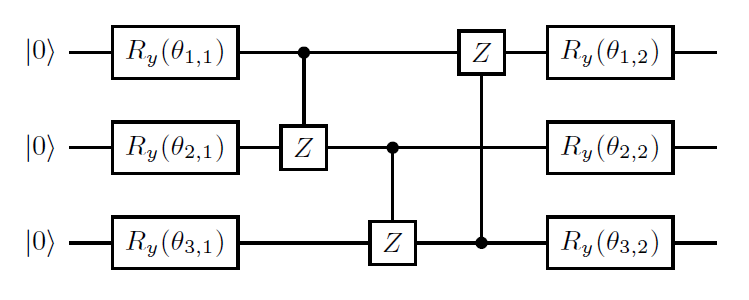}
    \caption{Parameterized family of gates $U(\boldsymbol{\theta})$, for three qubits, used for our simulated experiments. First, a layer of $R_y$ rotations is applied on each qubit, followed by a series of controlled-$Z$ operations applied in a nearest-neighbor fashion. Finally, another layer of $R_y$ rotations is applied on each qubit.}
    \label{fig:Ansatz_max_cut}
\end{figure}

\subsection{Classical Combinatorial Optimization Problems}
\label{subsec:classical_opt_problems}
\subsubsection{MaxCut}

The first problem that we considered is the \emph{MaxCut} problem. MaxCut is the most studied combinatorial optimization problem in the context of variational quantum algorithms literature \cite{wang2018quantum} due to its simple mapping to a Hamiltonian and its importance as a classical problem. The problem is defined as follows.

Let $G(V,E)$ be a non-directed $n$-vertex graph, where $V$ is the set of vertices, $E$ is the set of edges, and $w_{ij}$ are the weights of the edges. A \emph{cut} is defined as a bipartition of the set $V$ into two disjoint subsets $P,Q$, i.e. $P\cup Q = V$ and $ P\cap Q=\emptyset$. Equivalently, we label every vertex with either $0$ or $1$, where it is understood that the vertex belongs to set $P$ if it takes the value $0$ and to set $Q$ if it takes the value $1$. The aim is to maximize the cost function
 \begin{equation}
 C(\boldsymbol{x}) = \sum_{i,j = 1}^n w_{ij}x_i\left(1-x_j\right).
 \end{equation}
 Intuitively, this corresponds to finding a partition of the vertices into two disjoint sets that ``cuts'' the maximum number of edges. By transforming the binary variables $x_i$ to spin variables $z_i$ according to $x_i = \frac{1-z_i}{2}$, the cost function is mapped into a quantum spin-configuration problem,
 \begin{equation}
     C(\boldsymbol{z}) = \sum_{\left<i,j\right>\in E} \frac{w_{ij}}{2}\left(1-z_iz_j\right).
 \end{equation}
Maximizing the cost function above corresponds to finding the ground state of the Hamiltonian
\begin{equation}
	H_{\textrm{MC}} = -\sum_{\left<i,j\right>\in E}\frac{w_{ij}}{2}\left( 1- \sigma_{i}^z \sigma_j^z \right).
 \label{eq:maxcut_hamiltonian}
\end{equation}

\subsubsection{Number Partitioning}

 The second problem is the \emph{Number Partitioning} problem. Number Partitioning is a classical combinatorial optimization problem that has previously been tackled by VQAs \cite{nannicini2019performance, barkoutsos2020improving} and is stated as follows. We are given a set of $N$ integers $\{n_1, n_2, \ldots, n_N\}$ and we are asked to decide whether there exists a partition of the set into two disjoint subsets $S, \bar{S}$ so that the sums of the elements on each subset are equal. The above problem can be cast as an optimization problem with a cost function:
\begin{equation}
    C(\boldsymbol{x}) = \left(\sum_{i=1}^N (2x_i - 1)n_i\right)^2
\end{equation}
which can be readily formulated as a quantum spin problem interacting with Hamiltonian:
\begin{equation}
    H_{\textrm{NP}} = \sum_{i\neq j}(n_i n_j)\sigma_i \sigma_j + \sum_{i=1}^N n_i^2.
\end{equation}

\subsection{Transverse-Field Ising Chain}

The Transverse-Field Ising Chain (TFI) describes a quantum system that interacts under a Hamiltonian with extra off-diagonal terms. These types of Hamiltonians, compared to classical optimization Hamiltonians, have eigenvectors that do not correspond to the computational basis vectors. The performance of the TFI chain model has previously been investigated in \cite{wiersema2020exploring, meyer2023exploiting} using VQAs. The Hamiltonian describing the TFI chain model (with periodic boundary conditions) is:

\begin{equation}
    H_{\text{TFI}} = -\sum_{k=1}^n J_k \sigma^z_k \sigma^z_{k+1} - h \sum_{k=1}^n \sigma^x_k
\end{equation}
where ($J_k, h$) are coupling coefficients.

\subsection{Method Performance}
\subsubsection{Classical Optimization}

We start our analysis by presenting the results of AQC-PQC on MaxCut. We tested our method on 3-regular unweighted graphs of sizes 6 and 8 and performed 25 simulations for each size. For every instance, we discretized our algorithm into $K$ finite steps ranging from 2 to 26 and tested the performance of our strategy for each choice of steps. As a figure of merit, we chose the probability of sampling the optimal solution at the end of the algorithm as it is a well-suited metric for classical optimization problems (see \ref{subsec:metrics}).

In Figure \ref{fig:3-reg-graphs} we can visualize the average performance of our algorithm for different choices of discrete steps. For every choice of steps, we calculate the average performance of the instances with 6 qubits (orange line with pentagon marker) and with 8 qubits (blue line with diamond marker). It is clear that our algorithm performs quite well in the ideal case (noiseless executions of the quantum circuits and exact calculation of expectation values). Specifically, we observe a critical number of steps $S_C$ after which any choice of steps $S\geq S_C$ will return the optimal solution with probability one. 

\begin{figure}
\begin{tikzpicture}
\node (img) {\includegraphics[scale=0.35]{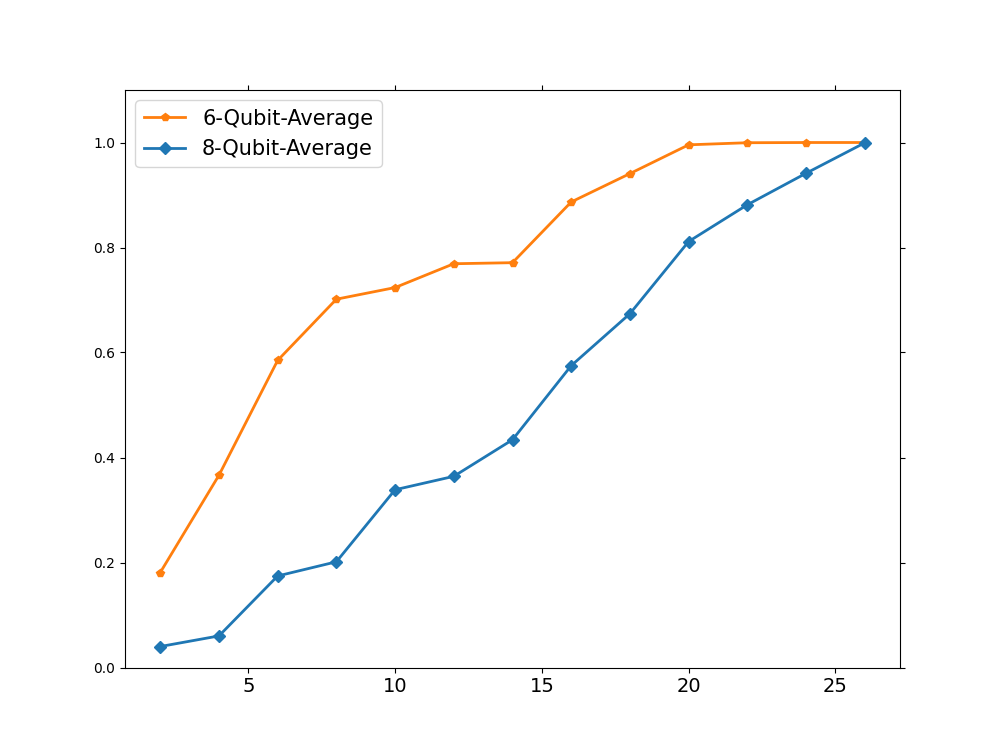}};
\node[below=of img,yshift=0.5cm, node distance=0cm, yshift=1.1cm] {\scriptsize Steps};
\node[left=of img, xshift=1cm, node distance=0cm, rotate=90, anchor=center,yshift=-0.7cm] {\scriptsize Probability of Optimal Solution};
\end{tikzpicture}
\caption{MaxCut. Average probability of optimal solution as a function of the choice of the discrete steps. For both sizes, the average probability increases with the number of steps, reaching over $80\%$ when the choice is 20 steps.}
\label{fig:3-reg-graphs}
\end{figure}

The results of the Number Partitioning problem are similar to those of MaxCut. However, AQC-PQC does not always converge to a quantum state with an overlap close to one. The reason is that for the Number Partitioning problem, a more expressive ansatz is needed so that at the intermediate timesteps, the system remains close to the ground state energy. We discuss this thoroughly in Section \ref{sec:comparisons} and in Appendix \ref{appendix:ansatz_expressiveness}.

\subsubsection{Transverse-Field Ising Chain}
\label{subsection:transverse_ising}

For our simulated experiments, we consider the random Transverse-Ising Chain Hamiltonian with periodic boundary conditions. We draw the coupling strengths $(J_k, h)$ at random from a random uniform distribution and we calculate the approximation ratio $\bra{\psi(\theta)}H_{\text{TFI}}\ket{\psi(\theta)}/E_{\text{exact}}$ for each choice of discretization steps at the end of the algorithm. For our numerical calculations, we used instances with sizes of 6,7 and 8 qubits, and discretization steps from 2 up to 20 steps. The results are illustrated in Figure \ref{fig:tfim}. We can see that for all instances the performance is similar. Specifically, the method can achieve near-optimal approximation ratios for a small number of steps. As the number of steps increases, the method reaches $0.95\%$ approximation ratio for all sizes and instances. 

\begin{figure}
\begin{tikzpicture}
\node (img)  {\includegraphics[scale=0.35]{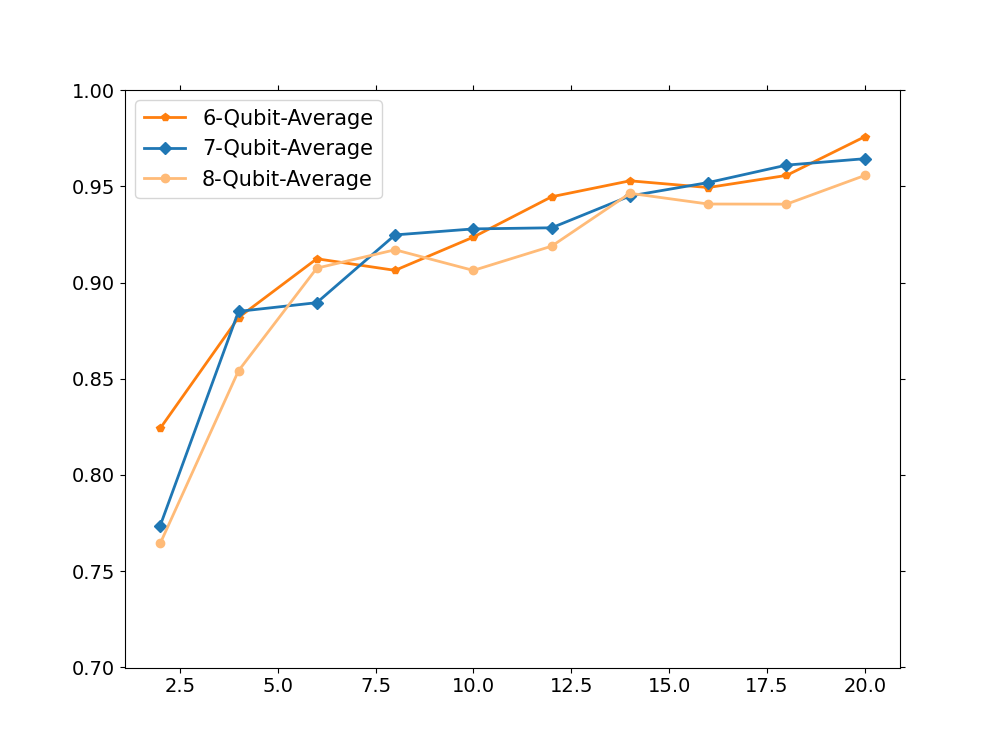}};
\node[below=of img,yshift=0.5cm, node distance=0cm, yshift=1.1cm] {\scriptsize Steps};
\node[left=of img, xshift=0.6cm, node distance=0cm, rotate=90, anchor=center,yshift=-0.7cm] {\scriptsize $\bra{\psi(\theta)}H_{\text{TFI}}\ket{\psi(\theta)}/E_{\text{exact}}$};
\end{tikzpicture}
\caption{Transverse-Field Ising Chains. Average approximation ratio as a function of the choice of discrete steps. Both instance sizes have similar performance achieving near-optimal approximation ratios.}
\label{fig:tfim}
\end{figure}

\section{Method Evaluation}
\label{sec:comparisons}
\subsection{Metrics and Technical Details}
\label{subsec:metrics}
In Section \ref{sec:experiments} we analyzed the performance of the AQC-PQC algorithm for two different problems with respect to the choice of discretization steps. However, it is important to examine how our method compares to VQE. As a figure of merit, we will use two different metrics. The first one is the distance from the optimal energy, $d = |E - E_{\text{opt}}|$, which is a good indicator of the quality of the output solution as it describes how far (in energy units) we are from the ground state energy. This metric is used both for the classical optimization problems as well as the quantum spin-configuration problem.

As we previously discussed, for classical combinatorial optimization problems, the ground state corresponds to a computational basis state. Thus, as a second figure of merit, we will use the overlap with the optimal solution which describes the probability of sampling the best solution. Specifically, let $\ket{\psi(\boldsymbol{\theta})}$ be a parameterized quantum state. The overlap of the state $\ket{\psi(\boldsymbol{\theta})}$ with a $d$-degenarate ground state is defined as:
\begin{equation}
    \sum_{i=1}^d \bra{\psi(\boldsymbol{\theta})}\ket{\psi_{\text{opt},i}}
\end{equation}
where $\ket{\psi_{\text{opt},i}}$ is the $i$-th optimal solution.

 Another thing to note is that AQC-PQC can be considered a \emph{second-order method} as it exploits information of both the gradients and the Hessian. To make a fair comparison, the classical optimization part in VQE will also make use of this information. Specifically, for the classical optimization part in VQE, we used both Gradient Descent (first-order optimization method) and 2-SPSA (second-order optimization method) with 10 random initialization points (in which we kept the best output out of 10). We evaluated our method in two different classical optimization problems and one quantum spin-configuration problem using the ansatz family seen in Fig \ref{fig:Ansatz_max_cut}. For the classical optimization problems, the ansatz family includes both the ground state of the initial Hamiltonian (which again is chosen to be the uniform superposition of all states) and the final Hamiltonian which is an (unknown) computational basis state while for the quantum spin-configuration problem, the ground state is not always included within the ansatz. 

 For AQQ-PQC, we discretized the Hamiltonian using 100 steps for all problems and instance sizes. Increasing the number of steps further would improve the performance of our method (as we discussed in Sec \ref{sec:experiments}), but we would like to state that even with a suboptimal choice of discretization steps, our method outperforms VQE even for a modest choice of ansatz family.

\subsection{Results}
\subsubsection{Classical Optimization Problems}

The two methods (VQE and AQC-PQC) were tested on the MaxCut and the Number Partitioning problems (details can be found in \ref{subsec:classical_opt_problems}). For these problems, we chose to compare the methods on instance classes that we consider hard. Both of these problems have an intrinsic $\mathbb{Z}_2$ symmetry and so we chose instances with only two optimal solutions (one solution can be acquired from the other by flipping all qubits). Specifically, for the MaxCut problem we sampled 100 \emph{random weighted graphs} of sizes 8 to 12 while for the Number Partitioning problem, we sampled 100 instances of the same size as MaxCut with integers drawn from the interval $[0, 50]$. The results are illustrated in Figure \ref{fig:performance_aqcpqc}. Additionally, we can see the overlap returned by both algorithms in Tables \ref{maxcut_table},\ref{number_partitioning_table}. 

Overall, we can see that AQC-PQC is able to outperform VQE in all instances, achieving overlap even five times larger in MaxCut (see Table \ref{maxcut_table}) and ten times larger in Number Partitioning (see Table \ref{number_partitioning_table}). Moreover, as seen in Figure \ref{fig:performance_aqcpqc}, the output states returned by AQC-PQC are significantly closer (in terms of energy) to the ground state compared to VQE. This is to be expected as the non-convexity of the cost landscape results in the classical optimization part of VQE to stuck in a local minimum. As pointed out in \cite{anschuetz2022quantum, you2021exponentially} the cost landscapes in shallow-depth VQAs, such as those utilized in this paper, are filled with a vast amount of local minima which makes them untrainable. One potential way out of this is to overparametrize the ansatz family \cite{larocca2023theory}, provided that the ansatz family can be overparametrized with a polynomial number of parameters. However, for NISQ devices, the large depth would result in noisy calculations due to the large number of errors and low coherence times.

On the other hand, AQC-PQC provides a more robust strategy to navigate the (time-evolving) landscape. Provided that the number of steps is chosen accordingly and the ansatz family is expressive enough, the latter algorithm will always achieve a large overlap with the optimal solution. 

However, it is important to stress that the expressiveness of the ansatz plays a significant role in the performance of the algorithm. We have observed that in the limit of very large steps, if the ansatz family is not sufficiently expressive, AQC-PQC will converge into a suboptimal solution. The reason is that for the intermediate ground states, circuits of large depth are required in order to remain close to the instantaneous ground state. In Appendix \ref{appendix:ansatz_expressiveness} we provide a case study in which we investigate the performance of AQC-PQC for different ansatz families.

\begin{figure*}[]
\begin{tikzpicture}
\node (img1)  {\includegraphics[scale=0.55]{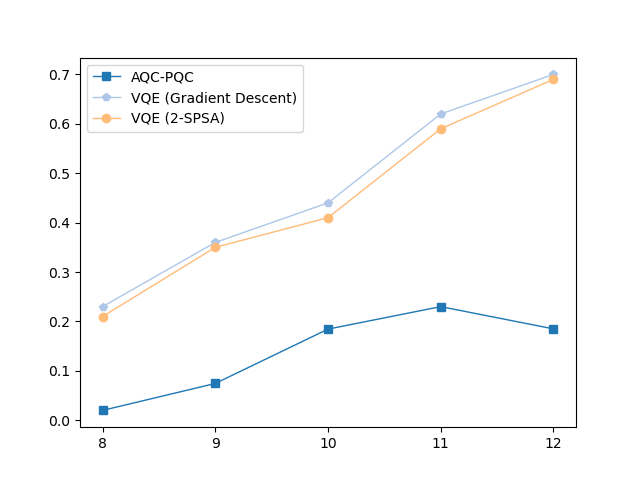}};
\node[below=of img1, node distance=0cm, yshift=1.5cm] (text1) {\scriptsize Number of Qubits};
\node[above=of img1, node distance=0cm, yshift=-1.8cm] {\textbf{MaxCut}};
\node[left=of img1, xshift=0.7cm, node distance=0cm, rotate=90, anchor=center,yshift=-0.7cm] {\scriptsize $|E-E_{opt}|$};
\node[right=of img1, xshift=-1cm] (img2)  {\includegraphics[scale=0.55]{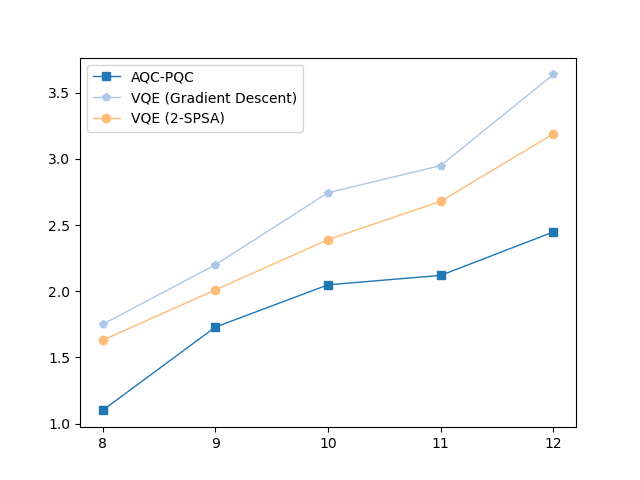}};
\node[below=of img2, node distance=0cm, yshift=1.5cm] {\scriptsize Number of Qubits};
\node[above=of img2, node distance=0cm, yshift=-1.8cm] {\textbf{Number Partitioning}};
\node[left=of img2, xshift=0.7cm, node distance=0cm, rotate=90, anchor=center, yshift=-0.7cm] {\scriptsize $|E-E_{opt}|$};
\end{tikzpicture}
\caption{Performance of AQC-PQC algorithm compared to VQE (with 2-SPSA and Gradient Descent optimizers) for the MaxCut problem (left) and the Number Partitioning problem (right). The AQC-PQC algorithm with the dark blue line (square markers) is able to outperform both 2-SPSA and Gradient Descent on the quality of the output solution.}
\label{fig:performance_aqcpqc}
\end{figure*}
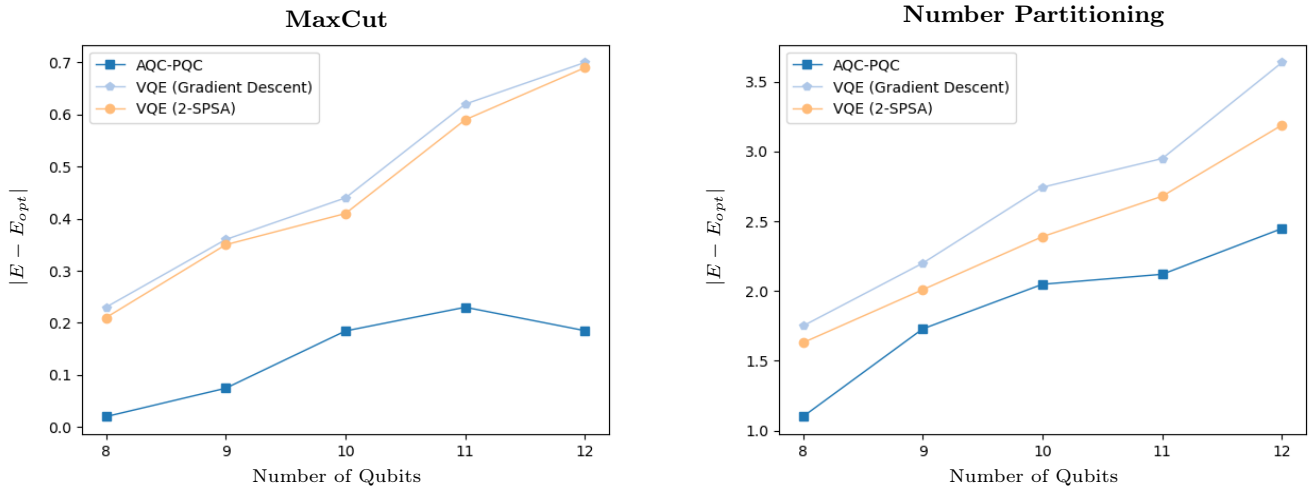

\begin{table*}
\begin{center}
\begin{tabular}{||c||c|c|c|c|c|c||}
    \hline
     \textbf{MaxCut} &\multicolumn{6}{c||}{Optimal Solution Overlap ($\%$)}\\
     \hline
      & \textbf{7 Qubits} & \textbf{8 Qubits} & \textbf{9 Qubits} & \textbf{10 Qubits} & \textbf{11 Qubits} & \textbf{12 Qubits}\\
     \hline
      AQC-PQC & 82.7 & 74.3 &93.1 & 50 & 28.1 & 56.6\\
     \hline 
      VQE& 62.3& 54.7& 60.8 & 39.2& 22.1& 11.1\\
     \hline
\end{tabular}
\caption{Probability of sampling the optimal solution for the MaxCut problem for instances of size 7-12. AQC-PQC was able to outperform VQE on all instance sizes achieving a larger overlap with the optimal solution.}
\label{maxcut_table}
\end{center}
\end{table*}

\begin{table*}
\begin{center}
\begin{tabular}{||c||c|c|c|c|c|c||}
    \hline
     \textbf{Number Partitioning} &\multicolumn{6}{c||}{Optimal Solution Overlap ($\%$)}\\
     \hline
      & \textbf{7 Qubits} & \textbf{8 Qubits} & \textbf{9 Qubits} & \textbf{10 Qubits} & \textbf{11 Qubits} & \textbf{12 Qubits}\\
     \hline
      AQC-PQC & 37.5 & 21.9 & 24.7 & 12.6 & 5 & 4.6\\
     \hline 
      VQE& 28.5& 6.2& 6.4 & 1.2& 0.8& 0.4\\
     \hline
\end{tabular}
\caption{Probability of sampling the optimal solution for the Number Partitioning problem for instances of size 7-12. AQC-PQC was able to achieve a significantly larger overlap than VQE for all instance sizes.}
\label{number_partitioning_table}
\end{center}
\end{table*}


\subsubsection{Transverse-Field Ising Chain}

Details for the TFI Chain problem can be found in \ref{subsection:transverse_ising}. We compared the two methods on instances of sizes 8 to 12 qubits. We evaluated the performance of the two algorithms on TFI Chain for 100 random instances with the couplings $(J_k, h)$ drawn uniformly at random from the uniform distribution. The results of the TFI Chain model are illustrated in Figure \ref{fig:results_tfim}. Overall, we observe that AQC-PQC is able to return approximations of the ground state of $H_{\text{TFI}}$ that are closer compared to those returned by VQE.

\begin{figure}
\begin{tikzpicture}
\node (img)  {\includegraphics[scale=0.55]{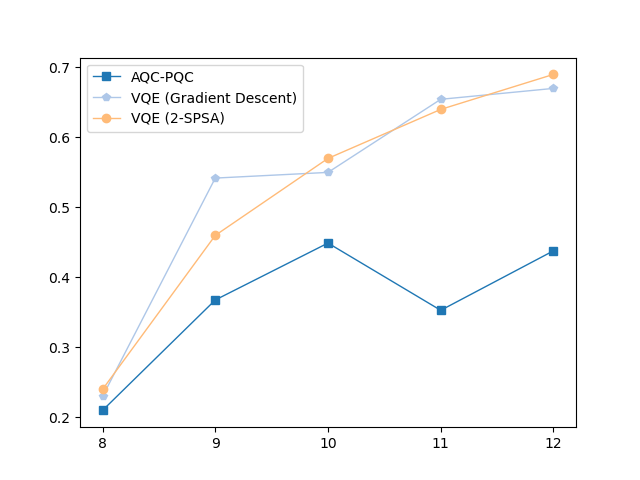}};
\node[below=of img, node distance=0cm, yshift=1.5cm] {\scriptsize Number of Qubits};
\node[left=of img, node distance=0cm, xshift=0.7cm, rotate=90, anchor=center,yshift=-0.7cm] {\scriptsize $|E-E_{opt}|$};
\end{tikzpicture}
\caption{Performance of AQC-PQC algorithm compared to VQE (with 2-SPSA and Gradient Descent optimizers) for the Transverse-Field Ising Chain model.}
\label{fig:results_tfim}
\end{figure}





\section{Discussion}
\label{sec:conclusion}

\subsection{Summary}

We introduced a new hybrid quantum-classical method that can be used to approximate adiabatic quantum computing (and thus universal quantum computing) using a parameterized family of quantum circuits suitable for early fault-tolerant devices. 

Our main mathematical result (Theorem \ref{th:main}) considered an ``initial'' Hamiltonian $H_s$, a final Hamiltonian $H_1$, a small parameter $\lambda\ll 1$, a family of quantum states $\ket{\psi(\ve{\theta})}=U(\ve{\theta})\ket{0}$ forming a parameterized quantum circuit (PQC), and the parameters $\ve{\theta^*}$ that give the state of the family with smallest initial energy. We then derived the shift vector $\ve{\epsilon}$ that the parameters need to be shifted by, to find the ground state of the perturbed Hamiltonian $H_s+\lambda V$. Importantly, we found that the shift vector $\ve{\epsilon}$ can be calculated by solving a linear system of equations which is derived using the expectations of some observables (energies and higher-order derivatives) estimated for the already known ground state of $H_s$.

This gives a new way to ``discretize'' adiabatic quantum computing. We choose the number of steps to be $K$, and starting from $H_0$ we (iteratively) perturb the Hamiltonian by $1/K(H_1-H_0)$ until it reaches the desired Hamiltonian $H_1$. We start with known parameters/state, and in each step, we estimate the expectations of the said observables at that point and evaluate the shift vector using Theorem \ref{th:main} to obtain the ground state of the next step. 

As a proof of principle, we checked that our method works for small instances of up to 12 qubits. We numerically tested our method (overlap with the optimal solution and distance from the optimal energy) on two different classical optimization problems, namely MaxCut and Number Partitioning, and also in a quantum-spin problem, the Transverse-Field Ising Chain model. We tested our algorithm on small instances of sizes 6 and 8 qubits and observed that in both cases, there exist a (small) critical number of steps, after which any choice will return the optimal solution with probability one. 

Moreover, we compared our method to the Variational Quantum Eigensolver (VQE) using two different classical optimization algorithms and 10 random initialization points. We observed that our method was able to outperform VQE both in terms of overlap with the optimal solution for the classical optimization problems but also in the distance from the optimal energy in all problems. The reason is that VQE is highly sensitive to the initialization points and the non-convexity of the cost landscape. As a result, a local optimization algorithm that follows the gradient will converge to the nearest local minimum provided that the learning step is sufficiently small. This means that if VQE is initialized near a local minimum then the output of the algorithm will be a sub-optimal solution. However, our method provides a more robust strategy in which we follow the optimal trajectory in the (evolving) landscape and reach (at the majority of times) the optimal solution provided a sufficient number of discretization steps and an expressive ansatz family.

\subsection{Evaluation, Comparisons and Future Research}

This approach opens many possibilities but also requires deeper theoretical and numerical work to fully understand its performance. The motivation of this paper was to introduce the technique and establish that it is working as expected and outline the potential it offers. Here we give some remarks on how we expect it to compare with other NISQ/early-fault-tolerant approaches (VQE, VAQC).

\noindent $\bullet$ Our simulated experiments involved small (but hard) instances of sizes up to 12 qubits in which VQE underperformed compared to AQC-PQC. We expect the difference in the performance of both algorithms to be significantly larger as the size of instances increases. However, we need to state that AQC-PQC would also require more expressive ansatz families and a larger number of steps to attain the same performance.

\noindent $\bullet$ In terms of resources, the comparison is not direct but we can note some advantages. Our cost -- the number of steps $K$, since we need to estimate a fixed number of expectations on specific states per step -- is well defined and essentially depends on the spectral gap (and the path from $H_0$ to $H_1$ that we choose). In VQE the running time (and number of states that one needs to prepare) depends on the number of iterations to convergence. This number can in practice only be estimated, and it is not clear how it scales with larger instances. Moreover, in small instances, one requires more iterations for the optimization than steps in our adiabatic evolution. At the same time, the classical problem we solve (the constrained set of linear equations) can be very efficiently solved with classical solvers, with guaranteed performance and good scalability.

\noindent $\bullet$ Perhaps the most important advantage is that the success of VQE and VAQC depends on the probability of converging to an actual global minimum. This inherits the challenges that classical optimization faces and is something that we avoid in our approach. For example, it is known that VQAs are sensitive to the initialization parameters. A bad initialization point may lead to false convergence, and as the size of problems increases, to phenomena like \emph{barren plateaux}, where the cost landscape is almost flat and exponentially many resources are required to train the parameters. These problems, however, do not persist in our case since we no longer initialize the parameters at random but at the ground state of the initial Hamiltonian and we do not perform any energy minimization at any step. Moreover, the cost landscapes of VQAs are filled with a vast number of local minima. Classical optimizers are known to struggle to find a near-optimal approximation in these non-convex energy landscapes. In our approach, we analytically calculate the optimal angles and essentially start with the global minimum. We then follow the path of that minimum as the Hamiltonian evolves and, provided the evolution is slow enough and the ansatz family expressive enough, we are \emph{guaranteed} to find a good approximation to the ground state. 

\noindent $\bullet$ Another potential advantage of our approach is that all the quantum states we prepare are close to the (time-evolved) ground state. There exist quantum error mitigation techniques that work better when the state/expectation value that needs to be recovered is a ground state \cite{bennewitz2022neural}. In contrast, in VQAs the quantum states prepared before the convergence are not close to being ground states, so VQA may benefit less from certain quantum error-correcting methods.

We can see (at least) three direct future directions:

\begin{enumerate}

    \item The time required for AQC is related to the number of steps $K$. While we provided evidence for the practicality of our method, an important task is to analyze $K$ extensively (both theoretically and using extensive numerical simulations). There are cases (ansatz families) where the first excited state may not necessarily correspond to a minimum and so the algorithm can acquire a speedup.
    
    \item Testing the method for different problems and (much) larger instances is the other obvious direction. This method, if combined with efficient classical solvers, should be suitable for running emulations with many qubits, and we expect to be able to see improvements in performance, compared to other variational quantum approaches for certain problems, already with $12-20$ qubits. This can include using knowledge from AQC (for example finding problems that require shorter time in AQC should be more promising, while another approach would be to consider other paths between $H_0$ and $H_1$ since for specific problems is known to give considerable advantages  \cite{roland2002quantum}).
    \item To address the issue of noise. Any real device is susceptible to imperfections. Our method will also be affected by noise. An obvious effect is that the expectation values used to compute the shift vector will be noisy, leading potentially to a wrong shift. Similarly, the accuracy with which one can tune the parameters $\ve{\theta}$ using a real physical device also affects the accuracy of our approach. This is analogous to ``noise-induced barren plateaux''~\cite{wang2021noise} since the noise essentially blurs the direction of the path that the parameters follow. While we anticipate the existence of such effects, it is important to explore the exact robustness of the method to imperfections, as well as the applicability and use of quantum error mitigation techniques in this context, with potential advantages compared to VQAs.
\end{enumerate}

\section*{Code Availability}
You can find a Python implementation of AQC-PQC using \href{https://github.com/ioankolot/AQC-PQC.git}{Qiskit} and a C++ implementation based on QuEST \cite{quest} as part of a \href{https://github.com/Milos9304/FastVQA}{FastVQA} library specifically designed for high-performance computing emulations.

\section*{Acknowledgements}

We thank Erika Andersson for useful discussions and comments. The authors acknowledge support from EPCC, including use of the NEXTGenIO system, which was funded
by the European Union’s Horizon 2020 Research and
Innovation programme under Grant Agreement no.
671951. IK and PW acknowledge partial funding from the ISCF grant 10001712, MP acknowledges support
by EPSRC DTP studentship grant EP/T517811/1 and PW acknowledges further support by EPSRC grants EP/T001062/1, EP/X026167/1 and EP/T026715/1, STFC grant ST/W006537/1 and Edinburgh-Rice Strategic Collaboration Awards. 

\bibliographystyle{unsrt}
\bibliography{References}

\onecolumngrid

\appendix

\section{Ansatz Expressiveness}
\label{appendix:ansatz_expressiveness}
It is very important to understand how crucial it is to have a parameterized family of gates that is sufficiently expressive. To be exact, the ansatz family should be expressive enough so that in the vicinity of small energy gaps, the energy returned by AQC-PQC is close to the ground state. In Figure \ref{fig:algorithm performance} we can visualize how close are the ground state energies returned by AQC-PQC at every step of the algorithm for different ansatz families.

Specifically, Figure \ref{fig:algorithm performance} illustrates a 3-regular graph of size 6 for the MaxCut problem. The system was initialized at the ground state of $H_0 = -\sum_i\sigma_i^x$ (with ground state $\ket{+}^{\otimes 6}$) and the Hamiltonian was discretized into 30 steps:
\begin{equation}
    H_k = \left(1-\frac{k}{30}\right)H_0 + \frac{k}{30}H_{\textrm{MC}}
\end{equation}
where $k\in [30]$ and $H_{\textrm{MC}}$ is the MaxCut Hamiltonian defined in Eq. \eqref{eq:maxcut_hamiltonian}. As a parameterized family of gates, we used the ansatz of Figure \ref{fig:Ansatz_max_cut} with 0 (no-entanglement gates), 1,2, and 3 layers of entanglement gates. Despite the fact that all these parameterized families contain both the initial and final ground states, they differ in the reachability of the intermediate ground states. We can see in Figure \ref{fig:algorithm performance} how the most expressive ansatz family (of 3 layers) is able to always remain close to the true ground state energy compared to the ansatz family which uses no entanglement. Although the latter achieves a non-zero overlap with the optimal solution, the final output energy is far from the optimal. As a result, for harder instances where the returned energy is larger than the first excited energy, there is a high probability that the resulting overlap will tend to zero.

\begin{figure}
\begin{tikzpicture}
\node (img) {\includegraphics[scale=0.38]{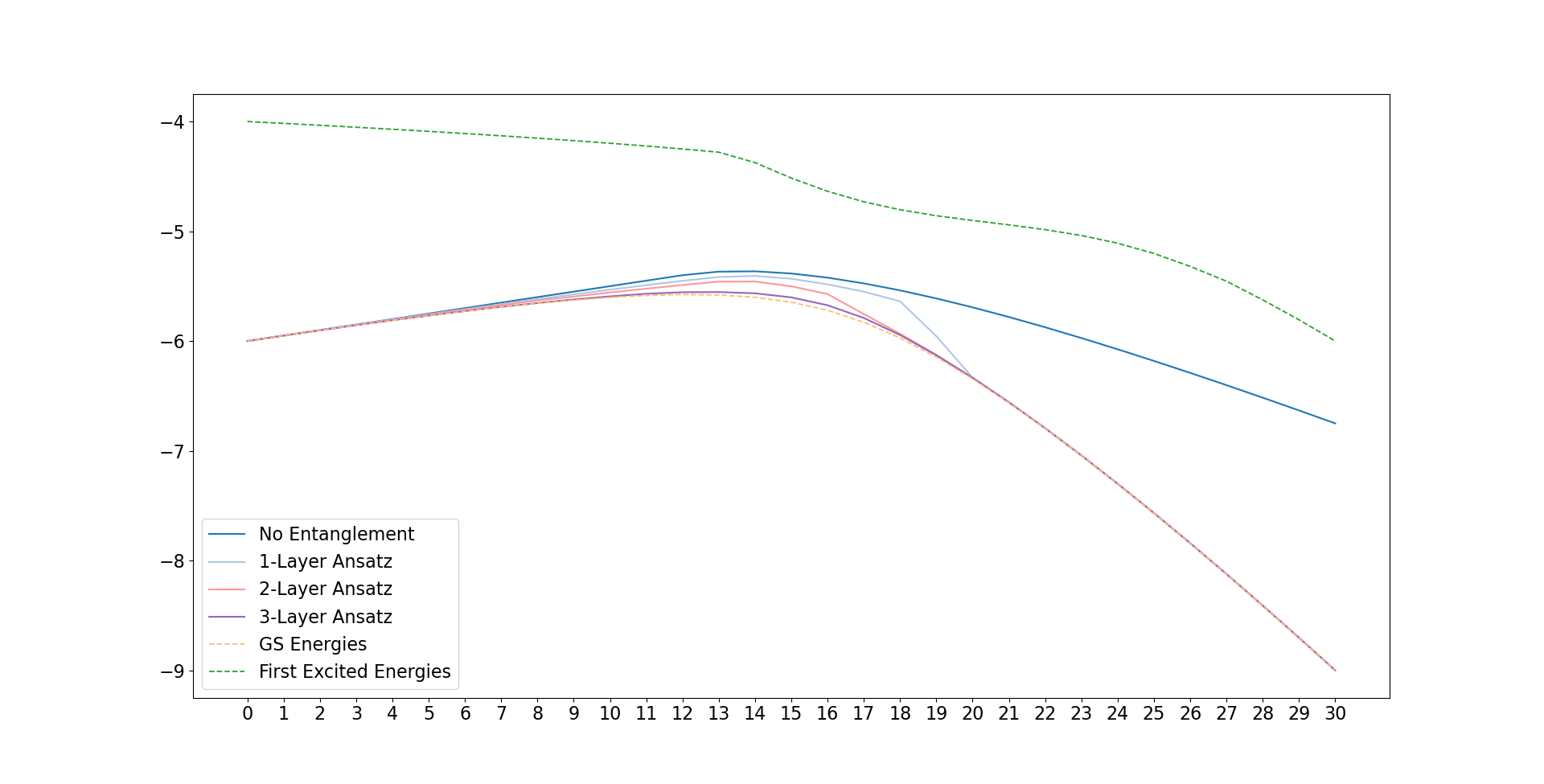}};
\node[below=of img,yshift=0.6cm, node distance=0cm, yshift=1.1cm] {\scriptsize Step};
\node[left=of img, xshift=2cm, node distance=0cm, rotate=90, anchor=center,yshift=-0.7cm] {\scriptsize Expectation Value};
\end{tikzpicture}
    \caption{AQC-PQC performance for three different parameterized family of gates for the MaxCut problem. The 3-Layer Ansatz (purple line) always remains close to the true ground state energy, returning a quantum state with an overlap equal to one at the end of the algorithm. On the other hand, the ansatz family that uses no entanglement (and thus is the least expressive) cannot remain close to the ground state (blue line) and so at the final step, it returns a solution that is far from the optimal.}
    \label{fig:algorithm performance}
\end{figure}

\section{Error Analysis}
\label{appendix:error_analysis}

\noindent For the analysis below (see \cite{folland2005higher} for details), we must first introduce some notation. Consider the multi-index $\boldsymbol{\alpha} = (\alpha_1, \alpha_2, \ldots, \alpha_n)$ where $\alpha_i$ are nonnegative integers. We define:
\begin{equation}
\begin{gathered}
    |\boldsymbol{\alpha}| = \alpha_1 + \alpha_2 + \ldots \alpha_n,\; \;
    \boldsymbol{\alpha}! = \alpha_1!\alpha_2!\ldots \alpha_n!\\
    \boldsymbol{x}^{\boldsymbol{\alpha}} = x_1^{\alpha_1}x_2^{\alpha_2}\ldots x_n^{\alpha_n} \: \text{for $x\in \mathbb{R}^n$},\; \;
    \partial^{\boldsymbol{\alpha}} f = \partial_1^{\alpha_1}\partial_2^{\alpha_2}\ldots \partial_n^{\alpha_n}f
\end{gathered}
\end{equation}

    \begin{theorem} (\emph{Taylor's Theorem for multivariable functions.} \cite{folland2005higher}) Let $f : \mathbb{R}^n\rightarrow \mathbb{R}$ be a $(k+1)$ differentiable function on an open convex set $S$. If $\boldsymbol{a}\in S$ and $\boldsymbol{a}+\boldsymbol{\epsilon}\in S$ then:
\begin{equation}
    f(\boldsymbol{a}+\boldsymbol{\epsilon}) = \sum_{|\boldsymbol{\alpha}|\leq k}\frac{\partial^{\boldsymbol{\alpha}} f(\boldsymbol{a})}{\boldsymbol{\alpha} !}\boldsymbol{\epsilon}^{\boldsymbol{\alpha}} + R_{\boldsymbol{a},k}(\boldsymbol{\epsilon})
\end{equation}
where the \emph{remainder} $R_{\boldsymbol{a},k}(\boldsymbol{\epsilon})$ is given in Lagrange's form by:
\begin{equation}
    R_{\boldsymbol{a},k}(\boldsymbol{\epsilon}) = \sum_{|\boldsymbol{\alpha}| = k+1} \partial^{\boldsymbol{\alpha}} f(\boldsymbol{a} + c\boldsymbol{\epsilon})\frac{\boldsymbol{\epsilon}^{\boldsymbol{\alpha}}}{\boldsymbol{\alpha}!}\: \: \text{for some $c\in \{0,1\}$}
\label{eq:remainder}
\end{equation}
\end{theorem}

\noindent One very useful property of Taylor's theorem is that it allows an analysis of the error when higher order terms are neglected. If bounds on the derivatives exist, then it is also possible to bound the remainder (i.e. the error in the approximation). Consider a Hamiltonian of the form:
\begin{equation}
    H = \sum_{l=0}^L c_l P_l
\end{equation}
comprised of $L=\mathcal{O}(\textrm{poly}(n))$ Pauli strings with $\norm{P_l} = 1$ and $c_l\in \mathbb{R}$. The expectation value of this Hamiltonian for every state $\ket{\psi}$ can be easily bounded as:
\begin{equation}
\begin{gathered}
|F| = |\bra{\psi}H\ket{\psi}| =
\Big{|}\bra{\psi}\sum_{l=0}^L c_l P_l\ket{\psi}\Big{|} \leq \\ \sum_{l=0}^L |c_l| |\bra{\psi} P_l\ket{\psi}| \leq \sum_{l=0}^L |c_l| \leq L \max_l |c_l|
\label{eq:exp_value_bound}
\end{gathered}
\end{equation}

\noindent In this paper we have considered PQCs that are comprised of parameterized gates that follow the parameter shift rules. However, the analysis can be easily extended (with some minor changes) to cases where the parameter shift rules do not hold. In our case, as we analyzed in \ref{sec:preliminaries} all derivatives can be written as a linear combination of expectation values in different parameter settings. As a result, all high-order derivatives can be bounded. Consider for example, the second order derivatives $\Big{|}\frac{\partial^2 F}{\partial \theta_i \partial \theta_j}\Big{|}$. Using the previous results (Eq. \eqref{eq:exp_value_bound}):
\begin{equation}
    \Big{|}\frac{\partial^2 F}{\partial \theta_i \partial \theta_j}\Big{|} \leq 4L\max_l |c_l|
\end{equation}

\noindent It is known \cite{folland2005higher} that if $f$ is $(k+1)$ times differentiable on a convex set $S$, with bounded derivatives $|\partial^{\boldsymbol{\alpha}} f(\boldsymbol{x})| < M$ and $|\boldsymbol{\alpha}| = k+1$ then:
\begin{equation}
    \Big{|}R_{\boldsymbol{a},k}(\boldsymbol{\epsilon})\Big{|} \leq \frac{M}{(k+1)!}\norm{\boldsymbol{\epsilon}}^{k+1}
\end{equation}
To see this, we have:
\begin{equation*}
    \Big{|}R_{\boldsymbol{a},k}\Big{|} = \sum_{|\boldsymbol{\alpha}|=k+1}\Big{|}\partial^{\boldsymbol{\alpha}} f\frac{\boldsymbol{\epsilon}^{\boldsymbol{\alpha}}}{\boldsymbol{\alpha} !}\Big{|}\leq M \sum_{|\boldsymbol{\alpha}|=k+1} \frac{\boldsymbol{|\epsilon}^{\boldsymbol{\alpha}}|}{\boldsymbol{\alpha} !} = \frac{M}{(k+1)!}\norm{\boldsymbol{\epsilon}}^{k+1}
\end{equation*}
where we used the \emph{multinomial theorem}:
\begin{equation*}
    (\epsilon_1 + \epsilon_2 + \ldots + \epsilon_n)^k = \sum_{|\boldsymbol{\alpha}| = k} \frac{k!}{\boldsymbol{\alpha}!}\boldsymbol{\epsilon}^{\boldsymbol{\alpha}}
\end{equation*}

\noindent We are interested in the case that we analyzed in Section \eqref{sec:param_theory}. We start with an initial Hamiltonian $H_0$:
\begin{equation}
    H_0 = \sum_{l=0}^{L_0} c_l P_l
\end{equation}
and introduce a small perturbation of the form $\lambda H_1$ with:
\begin{equation}
    H_1 = \sum_{l=0}^{L_1} b_l P_l
\end{equation}
We are interested in the error of the \emph{affine approximation} of the expectation value $F_\lambda$ of the perturbed Hamiltonian $H_\lambda = H_0 + \lambda H_1$. Using the notation that we introduced above, we can write $F_\lambda$ at the point $\boldsymbol{\theta^*} + \boldsymbol{\epsilon}$ as:
\begin{equation}
    F_\lambda (\boldsymbol{\theta^*} + \boldsymbol{\epsilon}) = F_\lambda(\boldsymbol{\theta^*}) + \sum_{i=1}^M \epsilon_i\frac{\partial }{\partial \theta_i} F_\lambda(\ve{\theta^*}) + R_{\boldsymbol{\theta^*}, 1}(\boldsymbol{\epsilon})
\label{eq:Taylor_exp_perturbed}
\end{equation}

\begin{corollary}
The remainder of the Taylor expansion of Eq. \eqref{eq:Taylor_exp_perturbed} when terms of order $k\geq 2$ are neglected is upper bounded as:
\begin{equation}
    |R_{\boldsymbol{\theta^*},1}(\boldsymbol{\epsilon})| \leq 2 \norm{\boldsymbol{\epsilon}}^2(L_0\max_j |c_j| + \lambda L_1 \max_l |b_l|)
\end{equation}
\end{corollary}

\noindent \emph{Proof.} The proof follows immediately from the analysis above.\\

\noindent As an immediate result, one can identify the convex set 
of feasible solutions $\boldsymbol{\epsilon}$ by allowing a desired tolerance $\tilde{\epsilon}$ in Eq. \eqref{eq:Taylor_exp_perturbed}. For example, if we want an error up to $\tilde{\epsilon} = 10^{-5}$ in Eq. \eqref{eq:Taylor_exp_perturbed} then:
\begin{equation*}
\begin{gathered}
        |R_{\boldsymbol{\theta^*},1}(\boldsymbol{\epsilon})| \leq 2 \norm{\boldsymbol{\epsilon}}^2(L_0\max_j |c_j| + \lambda L_1 \max_l |b_l|) \leq 10^{-5}\\
        \implies \norm{\boldsymbol{\epsilon}}^2 \leq \frac{10^{-5}}{2(L_0\max_j |c_j| + \lambda L_1 \max_l |b_l|)}
\end{gathered}
\end{equation*}

\end{document}